%% file: paper.tex
\documentclass[sigconf,authorversion]{acmart}
\usepackage{epsfig}
\usepackage{mathtools}
\usepackage{amssymb}
\usepackage{amsfonts}
\usepackage{amsmath}
\usepackage{amsthm}
\usepackage{ifthen}
\usepackage{enumerate}
\usepackage{latexsym}
\usepackage{graphicx}
\usepackage{textcomp}
\usepackage{color}
\usepackage[lofdepth,lotdepth]{subfig}
\usepackage{tablefootnote}
\usepackage{amssymb}
\usepackage{pifont}
\usepackage{tikz}
\usepackage{booktabs}
\usepackage{hyperref}
\usepackage{siunitx}
\usetikzlibrary{arrows,decorations.pathmorphing,backgrounds,positioning,fit,
    petri,automata,shadows,calendar,mindmap,positioning}

\usepackage{tabularx}
\usepackage{multirow}

\DeclareMathAlphabet{\mathcal}{OMS}{cmsy}{m}{n} 

\renewcommand{\Pr}{{\textup{Pr}}}
\renewcommand{\Pr}{\mathop{\bf Pr\/}}
\renewcommand{\P}{\mathop{\bf Pr\/}}

\newcommand{\Ex}{\mathop{\bf E\/}}
\newcommand{\xmark}{\ding{55}}%

\newcommand{\Var}{{\bf Var}}

\newcommand{\deq}{\mathrel{\mathop:}=}

\renewcommand{\phi}{\varphi}

\newcommand{\ignore}[1]{}
\newtheorem{attackgame}[theorem]{Attack Game}

\newtoggle{ccsversion}

\togglefalse{ccsversion}

\copyrightyear{2017} 
\acmYear{2017} 
\setcopyright{acmlicensed}
\acmConference{CCS '17}{October 30-November 3, 2017}{Dallas, TX, USA}
\acmPrice{15.00}\acmDOI{10.1145/3133956.3133989}
\acmISBN{978-1-4503-4946-8/17/10}

\iftoggle{ccsversion}{    
\settopmatter{printacmref=false, printfolios=false}
}{
    \setcopyright{none}
    \settopmatter{printccs=false, printacmref=false, printfolios=true}
}

\begin{document}
\title{T/Key: Second-Factor Authentication From Secure Hash Chains}

\author{Dmitry Kogan}
\affiliation{%
  \institution{Stanford University}
}
\authornote{Both authors contributed equally to the paper}
\email{dkogan@cs.stanford.edu}

\author{Nathan Manohar}
\affiliation{%
  \institution{Stanford University}
}
\authornotemark[1]
\email{nmanohar@cs.stanford.edu}

\author{Dan Boneh}
\affiliation{%
  \institution{Stanford University}
}
\email{dabo@cs.stanford.edu}


\begin{abstract}
Time-based one-time password (TOTP) systems in use today require
storing secrets on both the client and the server.  As a result, an
attack on the server can expose all second factors for all users in
the system.  We present T/Key, a time-based one-time password system
that requires no secrets on the server.  Our work modernizes the classic
S/Key system and addresses the challenges in making such a system
secure and practical. At the heart of our construction is a new lower bound
analyzing the hardness of inverting hash chains composed of independent 
random functions, which formalizes the security of this widely used primitive.
Additionally, we develop a near-optimal algorithm for quickly generating the 
required elements in a hash chain with little memory on the client.  
We report on our implementation of T/Key as an Android application. T/Key can be 
used as a replacement for current TOTP systems, and it remains secure in the
event of a server-side compromise.
The cost, as with S/Key, is that one-time passwords are longer 
than the standard six characters used in TOTP.  
\end{abstract}

\iftoggle{ccsversion}{        

\ccsdesc[500]{Security and privacy~Hash functions and message authentication 
codes}
\ccsdesc[500]{Security and privacy~Multi-factor authentication}

\keywords{two-factor authentication, hash chains} 
}

\maketitle

\section{Introduction}
\input{intro}


\section{Offline 2nd Factor Authentication}
\label{TFA}

\ignore{
Static passwords have been notorious for their insecurity for almost four
decades~\cite{morris1979password, bonneau2010password, wiki:icloudhack, 
wiki:amhack, wiki:yahoohack}. Due to these issues, companies as well as web 
providers 
have been increasingly adopting two-factor authentication schemes, schemes in 
which users 
need to provide more than one ``secret"~\cite{rsasecurid,google2sv}. In such 
schemes, the first factor is 
usually referred to as ``something that you know'' (typically, static 
passwords are still used in this step), with the second factor being a proof of 
``something that you have'' such as a 
designated hardware token~\cite{rsasecurid, u2f} or a mobile 
phone~\cite{wu2004secure}. The possession of such a device 
is ``proved" by providing a one-time password (OTP) generated by the device 
(possibly in response to some challenge) or delivered to it over an external 
communication channel (such as an SMS message or an Internet connection). 
}

We begin by briefly reviewing several approaches to one-time passwords
that are most relevant to our scheme.

\paragraph{S/Key}
The idea of a one-time password authentication scheme was first considered 
by Lamport~\cite{lamport81}. 
Loosely speaking, in such a scheme, following an initial setup phase, 
authentication is performed by the client presenting the server with a password 
that is hard for an attacker to guess, even given all previous communication 
between the server and the client. In particular, no password is valid for more 
than one authentication. 
In his work, Lamport proposed a concrete instantiation of this idea using 
\emph{hash chains}, and this idea has been subsequently developed and 
implemented under the name S/Key~\cite{skey}. The setup phase of S/Key consists 
of the client choosing a secret passphrase\footnote{Usually, the client's 
secret passphrase is concatenated with a random salt to prevent dictionary 
attacks and reduce the risk of reusing the same passphrase on multiple 
servers.} $x$ and sending the computed value $y_0=h^{(k)}(x)$ 
(where $h$ is some cryptographic hash function, $k$ is some integer, and 
$h^{(k)}$ denotes 
$k$ successive iterations of $h$) to the server, which the server then stores. 
Subsequently, to authenticate for the $i$th time, the client must present 
the server with $y_i=h^{(k-i)}(x)$, which the server can verify by computing 
$h(y_i)$ and comparing it to the stored value $y_{i-1}$. If the 
authentication 
is successful, the server updates its stored value to~$y_i$. 

S/Key has a number of undesirable properties. First, one-time passwords remain 
valid for an indefinite period of time unless used, making them vulnerable to 
theft and abuse. This vulnerability is magnified if the counter value for each
authentication attempt is communicated to the client by the server, as is the
case in both the original S/Key~\cite{skey} and in the newer OPIE~\cite{opie}
(presumably to allow for stateless clients).  In this common setting, the scheme
is vulnerable to a so-called ``small $n$" attack~\cite{mitchell96}, where an
attacker impersonating the server can cause the client to reveal a future
one-time password. Second, the fact that S/Key utilizes the same hash function
at every iteration in the chain makes it easier to break S/Key than
to break a single hash function (see Theorem~\ref{thm:hastad}). This
also implies that any modification to the scheme that requires using much longer
hash chains (such as, for example, a na\"ive introduction of time-based
passwords) could lead to insecurity.

\paragraph{HOTP}
In an \emph{HMAC-based one-time password scheme} (HOTP)~\cite{hotp}, 
a secret and a
counter, both shared between the server and the client, are used in conjunction
with a pseudorandom function (HMAC) to generate one-time passwords. 
The setup phase consists of the server and the client agreeing on
a random shared secret $k$ and initializing a counter value $c$ to $0$.
One-time passwords are then generated as $\texttt{HMAC}(k,c)$. The counter is
incremented by the client every time a password is generated and by the server
after every successful authentication. 

The most significant advantage of this scheme is that the number of
authentications is unbounded.  Moreover, it allows using short
one-time passwords without compromising security.  However, HOTP
still suffers from many of the weaknesses of S/Key, namely that
unused passwords remain valid for an indefinite period of time. 
A bigger concern is that the secret key $k$ must be stored on the
server, as discussed in the previous section.

\paragraph{TOTP}
Time-based one-time password schemes (TOTP)~\cite{totp} were introduced to 
limit the validity period of one-time passwords. 
In TOTP, the shared counter value used 
by HOTP is replaced by a timestamp-based 
counter. Specifically, the setup phase consists of the server and the client
agreeing on the `initial time' $t_0$ (usually the UNIX epoch) and a time 
slot size $I$ (usually 30 seconds), as well as on a secret key $k$. 
Subsequently, 
the client can authenticate by computing $\texttt{HMAC}(k,(t-t_0)/I)$, 
where $t$ is the time of authentication. 
\ignore{
To avoid failures resulting from delays 
that may naturally occur in the process (such as the time it takes the user to 
type-in the one-time password or the transmission delay), the server is advised 
to check the provided password against all past time-stamps that fall within 
some validity window (which is usually set to one time-slot~$I$). Since the 
one-time passwords are only valid during the corresponding time window, TOTP 
solves the weakness of indefinitely valid unused passwords present in the HOTP 
and S/Key schemes.} 
As with HOTP, the TOTP scheme is vulnerable to a server-side attack. 

\ignore{  
\paragraph{Public Key Schemes}
Authentication schemes based on public-key cryptography~\cite{rsa} succeed in 
not requiring any secret information to be stored by the verifier. 
However, while  such schemes have reached universal usage for 
server authentication~\cite{tlssecurity}, client-side certificates have gained 
limited deployment~\cite{clientcertadoption} for 
authenticating users on the web, due to a more complicated
user-experience~\cite{herley2009passwords} and a lack of portability between 
devices~\cite{seckeys}. Moreover, the relatively high length of digital 
signatures (the shortest signatures at a 128-bit security level are 384 bits 
long~\cite{bernstein2012high}) makes them difficult for manual entry.
}

\section{Our Construction} \label{sec:description}

T/Key combines the ideas used in S/Key and TOTP to achieve the best
properties of both schemes: T/Key stores no secrets on the
server and ensures that passwords are only valid for a short time
interval. The scheme works as follows:

\begin{figure*}
\centering
\resizebox{16cm}{!}{%
\begin{tikzpicture}
[place/.style={circle,draw,fill=blue!20,thick,
                 inner sep=0pt,minimum size=11mm}]

\node at (0,0)  [place]  (Pinit) {$p_{\rm init}$};
\node at (2.5,0)  [place]  (P1) {$p_1$};
\node at (5,0)  [place]  (P2) {$p_2$};
\phantom{\node at (7.5,0)  [place] (P3) {};}
\phantom{\node at (10,0) [place] (Pkm2) {};}
\node at (12.5,0) [place] (Pkm1) {$p_{k-1}$};
\node at (15,0) [place] (Pk) [label=below:head] {$p_k$};

\phantom{\node at (-.8, -3) [place] (T0) {};}
\phantom{\node at (15.8, -3) [place] (Tk) {};}
\node at (-0.25, -3.5) (tinit) {$t_{\rm init}$};
\node at (15.25, -3.5) (tmax) {$t_{\rm max}$};
\node at (5, -3.5) (tprev) {$t_{\rm prev}$};
\node at (12.5, -3.5) (tmax1) {$t_{\rm max} - 1$};

\node at (12.5,3) [place,fill=pink] (P) {$p$} edge [->,>=stealth',semithick,bend 
right=30] (P2);

\node at (8.75,2.25) [rotate=20] {$h_{k-2} \circ \cdots \circ h_2$} ;
\node at (8.25, 3.5) [rotate = 20] {?};
\node at (8.375, 3.2) [rotate = 20] {=};

\node at (5,-2) (Pprev) {$p_{\rm prev}$};

\draw (9,0) node [align=center] {$\cdots$}; 

\draw (1.25,0.4) node [align=center]  {$h_k$};
\draw (3.75,0.4) node [align=center]  {$h_{k-1}$};
\draw (6.25,0.4) node [align=center]  {$h_{k-2}$};
\draw (11.25,0.4) node [align=center]  {$h_2$};
\draw (13.75,0.4) node [align=center]  {$h_1$};

\draw[->,>=stealth',semithick] (P1) -- (Pinit) ;
\draw[->,>=stealth',semithick] (P2) -- (P1) ;
\draw[->,>=stealth',semithick] (P3) -- (P2) ;
\draw[->,>=stealth',semithick] (Pk) -- (Pkm1) ;
\draw[->,>=stealth',semithick] (Pkm1) -- (Pkm2) ;

\draw[->,>=stealth',semithick] (Pprev) -- (P2) ;

\draw[->,>=stealth',semithick] (T0) -- (Tk);

\draw[-,>=stealth',semithick] (-0.25,-3.25) -- (-0.25,-2.75);
\draw[-,>=stealth',semithick] (5, -3.25) -- (5, -2.75);
\draw[-,>=stealth',semithick] (12.5, -3.25) -- (12.5, -2.75);
\draw[-,>=stealth',semithick] (15.25, -3.25) -- (15.25, -2.75);

\usetikzlibrary{decorations.pathreplacing}
\draw[decorate, decoration={brace, amplitude=10pt}] (12.5, -2.5) -- (15.25, 
-2.5);
\node at (13.925, -1.75) (I) {$I \text{ seconds}$};

\draw[dashed, semithick] (12.5, 0.75) -- (12.5, 2.25);

\end{tikzpicture}
}

\caption{ A basic overview of T/Key.
A user has submitted the password $p$ at time $t_{\mathrm{max}} - 1$. Since the
previous login occurred at time $t_{\mathrm{prev}} = t_{\mathrm{init}} + 2$, the
server has stored $p_{\mathrm{prev}} = p_2$ as the previous password. To
authenticate the user, the server computes $h_{k-2}(\hdots(h_2(p))\hdots)$ and
checks if it is equal to $p_{\mathrm{prev}}
$. 
}\label{scheme:overview}
\end{figure*}

\paragraph{Public Parameters} The scheme's parameters are the password 
length $n$ (in bits), a time slot size $I$ (in seconds), representing the amount of time each password is valid for, and the maximal supported authentication period $k$ (measured
as the number of slots of size $I$). Furthermore, our scheme uses some public
cryptographic hash function ${H:\{0,1\}^{m}\rightarrow 
\{0,1\}^m}$ for an arbitrary $m\ge n + s + c$, where $s$ is the number of bits
used for the salt and $c$ is the number of bits needed to represent the time.
Typical values are given in Table~\ref{table:params}.

\begin{table}[!h]
    \caption{Scheme public parameters and their typical values.}
    \label{table:params}
    \begin{tabular}{@{}
        c
        c
        l}
        \toprule
         {Parameter} & {Value} & {Description} \\
        \midrule
         $n$ & $130$ bits & One-time password length \\
         $s$ & $80$ bits & Salt length \\
         $c$ & $32$ bits & Number of bits used for time \\
         $m$ & $256$ bits & Hash function block size \\
         $k$ & $2\times 10^6$ & Chain length \\
         $I$ & $30$ sec & Time slot length \\
        \bottomrule
    \end{tabular}%
\end{table}

\paragraph{Setup} The client chooses and stores a uniformly random secret key
$\mathit{sk} \in \{0,1\}^n$, as well as a random salt $\mathit{id}\in
\{0,1\}^s$, and notes the setup time $t_{\mathrm{init}}$ (measured in slots of
length $I$). The public hash function $H$ together with the initialization time
$t_{\mathrm{init}}$ induce the $k$ independent hash functions $h_1,\dotsc,h_k
:\{0,1\} ^n\rightarrow \{0,1\}^n$ as follows: for $1 \leq i \leq k$ define
\begin{align*}
{h_i(x) = H\big(\langle t_{\mathrm{init}} + k - i \rangle_c \ \big\| \
\mathit{id} \ \big\| \  x\big)\Big\vert_n} \,,
\end{align*}
where for a numerical value $t$, $\langle t\rangle_c$ denotes the $c$-bit binary
representation of $t$, and for strings $x,y\in\{0,1\}^*$, we write $x\vert_n$
and $x\| y$ to denote the $n$-bit prefix of $x$ and the concatenation of $x$ and
$y$, respectively.  This simple method of obtaining independent hash functions
from a single hash function over a larger domain is called \emph{domain
separation}, and it is often attributed to Leighton and
Micali~\cite{leightonmicali}. Note that since all inputs to the hash function
are of equal size, this construction is not susceptible to length extension
attacks, and therefore, there is no need to use HMAC.

The client then computes 
\begin{align*}
p_{\mathrm{init}} = h_k(h_{k- 1}( \hdots (h_{1}(\mathit{sk})) \hdots))
\end{align*}
and sends it to the server together with $\mathit{id}$. 
The server stores $p_{\mathrm{init}}$ as $p_{\mathrm{prev}}$ as well as the time
$t_{\mathrm{init}}$ as $t_{\mathrm{prev}}$ (we discuss time synchronization
issues below). 

\paragraph{Authentication} To authenticate at a later time $t \in
(t_{\mathrm{init}}, t_{\mathrm{max}}]$ (measured in units of length $I$ where
$t_{\mathrm{max}} = t_{\mathrm{init}} + k$), the client and server proceed as
follows: the client uses $\mathit{sk}$ and $t$ to generate the
one-time password 
\[ 
    p_t = h_{t_{\mathrm{max}}-t}(h_{t_{\mathrm{max}}-t-1}(\hdots
    (h_{1}(\mathit{sk}))\hdots)).
\]
Alternatively, when $t=t_{\mathrm{max}}$, we use ${p_t= \mathit{sk}}$.
To check a password $p$, the server uses the stored
values, $t_{\mathrm{prev}}$ and $p_{\mathrm{prev}}$, and the current time-based
counter value $t > t_{\mathrm{prev}}$. The server computes 
\begin{align*}
p'_{\mathrm{prev}} =
      h_{t_{\mathrm{max}}-t_{\mathrm{prev}}}(
          h_{t_{\mathrm{max}}-t_{\mathrm{prev}}-1}(\dotsc
      (h_{t_{\mathrm{max}} - t + 1}(p)) \hdots)).
\end{align*}
If $p'_{\mathrm{prev}} = p_{\mathrm{prev}}$, then authentication is successful,
and the server updates $p_{\mathrm{prev}}$ to $p$ and $t_{\mathrm{prev}}$ to
$t$. Otherwise, the server rejects the password.

\paragraph{Reinitialization}
Just as in authentication, initialization requires communication only from the 
client device to the server, and the server does not need to send anything to 
the client. The only difference is that during initialization, the client needs 
to supply the server with the salt in addition to the initial password. 
The finite length of the hash chain requires periodic reinitialization, and the 
length of this period trades off with the time step length $I$ and the time it 
takes to perform the initialization (which is dominated by the full traversal of
the hash chain by the client). For standard use cases, one can set $I=30$ 
seconds and $k=2\times 10^6$, which results in a hash chain valid for $2$ years
and takes less than $15$ seconds to initialize on a modern phone.

Since key rotation is generally recommended for security purposes (NIST, for 
example, recommends ``cryptoperiods'' of 1-2 years for private authentication 
keys~\cite{nistkeys}), we don't view periodic reinitialization as a
major limitation of our scheme. While reinitialization is obviously somewhat 
cumbersome, there are several properties of our scheme that mitigate the 
inconvenience. First, the fact that our setup is unidirectional makes it very
similar to authentication from the user's point of view. Second, from a
security standpoint, the setup is not vulnerable to passive eavesdrop
attacks, unlike TOTP schemes that rely on shared secrets.

A scenario where the hash chain expires before the user is able to
reinitialize it with the server can be handled out-of-band in a manner similar to password
recovery or loss of the second-factor. Alternatively, some implementations could
choose to accept the head of the chain even after its validity period, which would incur a
loss in security proportional to the time elapsed since expiration.

\paragraph{Clock Synchronization.} As with current TOTP schemes, authentication
requires a synchronized clock between the server and the client. Time skew, or
simply natural delay between the moment of password generation and the moment of
verification, might result in authentication failure. To prevent this, the
server may allow the provided password to be validated against several previous
time steps (relative to the server's clock), as was the case in the TOTP scheme.
When this occurs, the previous authentication timestamp $t_{\mathrm{prev}}$
stored on the server should be updated to the timestamp which resulted in
successful verification.

Figure~\ref{scheme:overview} illustrates the design of T/Key.


\section{Security}
\label{sec:security}

Although our scheme bears a resemblance to both S/Key and TOTP, it has several 
essential differences that eliminate security issues present in those schemes. 

First and foremost, T/Key does not require the server to store any 
secrets, which mitigates the risk of an attack that compromises the server's 
database, unlike TOTP, which requires the client's secret key to be stored by 
the server. 

Second, T/Key's passwords are time limited, unlike those in S/Key, 
which makes phishing attacks more difficult because the attacker has a limited 
time window in which to use the stolen password. However, the fact that T/Key's 
passwords are time limited makes it necessary for the hash chain used by T/Key 
to be significantly longer than those in S/Key, since its length must now be 
proportional to the total time of operation rather than to the supported number 
of authentications. This modification raises the issue of the dependence of 
security on the length of the hash chain. 
Hu, Jakobsson and Perrig~\cite{HJP} discuss the susceptibility of iterating 
the same hash function to ``birthday" attacks and H{\aa}stad and 
N{\"a}slund~\cite{hastad} show that if the \emph{same} hash function $h$
is used in every step of the chain, then inverting the $k$-th iterate is
actually $k$ times easier than inverting a single instance of the hash 
function. We reproduce their proof here for completeness and clarity.

We set $N=2^n$ and denote by $\mathcal{F}_N$ the uniform distribution over the
set of all functions from $[N]$ to $[N]$. For a function ${h:[N]\rightarrow
[N]}$, we let $h^{(k)}$ denote $h$ composed with itself $k$ times. For functions
$h_1,h_2,\dotsc,h_k$ and $1\le i \le j \le k$, we let $h_{[i,j]}$ denote the
composition $h_j\circ h_{j-1}\circ \dotsb \circ h_i$. When writing $A^h$, we
mean that algorithm $A$ is given oracle access to all $k$ functions
$h_1,\dotsc,h_k$. 

\begin{theorem}[\cite{hastad}] \label{thm:hastad}
    For every $N\in\mathbb N$, $k\leq \sqrt{N}$ and $2k\leq T \leq N/k$, there
    exists an algorithm $A$ that makes at most $T$ oracle queries to a random
    function $h:[N]\rightarrow[N]$ and
    \[
\Pr_{\substack{h\in\mathcal{F}_N \\ x\in[N]}}\left[h\left(A^h(h^{(k)}(x))\right)
= h^{(k)}(x)  \right] = \Omega\left(\frac{Tk}N\right)\,.
    \] Moreover, every algorithm that makes at most $T$ oracle queries succeeds
    with probability at most $O(Tk/n)$.
\end{theorem} 
\begin{proof}
    We prove the first part of the theorem (the existence of a ``good''
    algorithm) and refer the reader to~\cite{hastad} for the proof of the second
    part. Consider the following algorithm: On input $h^{(k)}(x) = y\in[N]$, the algorithm
    sets $x_0=y$ and then computes $x_j=h(x_{j-1})$ until either $x_j=y$, in
    which case it outputs $x_{j-1}$, or until $x_j=x_i$ for some $i<j$. In the
    latter case, it picks a new random $x_j$ from the set of all points it
    hasn't seen before and continues. If the algorithm makes $T$ queries to $h$
    without finding a preimage, it aborts. 

    To analyze the success probability of this algorithm, consider the first
    $(T-k)$ points $\{x_j\}_{j=1}^{T-k}$. If any of these points collides with
    any of the values along the hash chain $\{h^{(i)}(x)\}_{i=1}^k$, the algorithm
    will output a preimage of $y$ after at most $k$ additional queries.
    Therefore, the probability of failure is at most the probability of not
    colliding with the hash chain during the first $T-k$ queries. But as long as
    a collision does not happen, each query reply is independent of all previous
    replies and of the values $\{h^{(i)}(x)\}_{i=1}^k$. Each query therefore
    collides with the chain with probability at most $k/N$, and overall, the
    algorithm fails with probability at most $(1-k/N)^{T-k} \leq  (1-k/N)^{T/2}
    = 1-\Omega\left(Tk/N\right)$.
\end{proof} 

This loss of a multiplicative factor of $k$ in security is undesirable as it
forces us to increase the security parameters for the hash function to resist
long-running adversaries. A standard solution is to use a different hash
function at every step in the chain.  The question then is the following:  if
$H$ is the composition of $k$ {\em random} hash functions, namely
\[  H(x) \deq h_k(h_{k-1}(\cdots(h_2(h_1(x))) \cdots ))\,,  \] how difficult is
it
to invert $H$ given $H(x)$ for a random $x$ in the domain? To the best of our
knowledge, this aspect of hash chain security has not been analyzed previously.

In Section~\ref{sec:online} we prove a time lower bound for inverting a
hash chain composed of independent hash functions. We show that as opposed to
the case in Theorem~\ref{thm:hastad}, where the same function is used throughout
the chain, resulting in a loss of security by a factor of $O(k)$, using independent
hash function results in a loss of only a factor of $2$. Thus for most practical
applications, a hash chain is as hard to invert as a single hash function. 
In Section~\ref{sec:preprocessing}, we prove a time-space lower bound for
inverters that can preprocess the hash function. 

\subsection{A lower bound for inverting hash chains}
\label{sec:online}

\begin{theorem}[Security of hash chains against online attacks]
    \label{thm:onlineinvert}
    Let functions ${h_1,\dotsc,h_k \in [N]\rightarrow [N]}$ be chosen 
    independently and uniformly at random. Let $A$ be an algorithm that 
    gets oracle access to each of the functions $\{h_i\}_{i=1}^k$ and makes at 
    most 
    $T$ oracle queries overall. Then, 
    \begin{align*}
    \P_{\substack{h_1,\dotsc,h_k \in \mathcal F_N\\ x_0 \in [N]}}
    \left[h_k\left(A(h_{[1,k]}
    (x_0))\right)  = h_{[1,k]}(x_0)\right] \le \frac {2T+3}{N}.
    \end{align*}
\end{theorem}

\begin{proof}
    
    Let $W=(w_0,w_1,\dotsc,w_k)$ be the sequence of values of the hash chain, 
    i.e., $w_0=x_0$ and $w_i=h_i(w_{i-1})$ for $i\in[1,k]$.
    Let $A$ be an adversary that makes at most $T$ oracle queries.  
    Denote by $q_j=(i_j,x_j,y_j)$ the $j$-th query 
    made by $A$, where $i_j$ is the index of the oracle queried, $x_j$ is the 
    input queried, and $y_j$ is the oracle's response. We say that a query 
    $q_j$ \emph{collides} with $W$ if $y_j 
    =  w_{i_j}$, namely the reply to the query is a point on the hash chain.
    At the cost of one additional query, we modify $A$ to query $h_k$ on its
    output before returning it. Thus, we can assume that if $A$ successfully
    find a preimage, at least one of its $T+1$ queries collides with $W$.

    Let $R=\{(i,x,y) : h_i(x)=y\}$ be the set of all random oracle queries and 
    their 
    answers. Using the principle of deferred decision, we can construct the set 
    $R$ 
    incrementally as follows. Initially $R=\emptyset$; 
    subsequently whenever $A$ makes an oracle query of the form $(i,x)$, if 
    $x=w_{i-1}$, we respond with $y=w_i$ and add $(i, w_{i-1}, w_i)$ to $R$. 
    Else if $(i,x,y)\in R$, we reply with  $y$. 
    Otherwise, we choose $y$ uniformly at random from $[N]$, add 
    $(i,x,y)$ to 
    $R$, and reply with $y$. 
    
    As mentioned above, to invert the hash chain, at least one query $q_j\in R$ 
    must 
    collide with $W$. It follows that
    \begin{align*}
    \Pr_{H, x_0}\left[A \text{ loses}\right] = 
    &\Pr_{H, x_0}\left[\bigwedge_{j=1}^{T+1} y_j\neq w_{i_j}\right] \\
    = &\prod_{j=1}^{T+1}\Pr_{H, x_0}\left[y_j\neq w_{i_j} \middle | 
    \bigwedge_{\ell=1}^{j-1}y_\ell\neq w_{i_\ell}\right]. 
    \end{align*}
    
    To bound each term inside the product, we use the basic fact that
    \begin{align*}
    \P(A|C) &= \P(A|B,C)\P(B|C) + \P(A|\neg B,C)\P(\neg B|C) \\
    &\le \P(A|B,C) + \P(\neg B|C)
    \end{align*}
    to obtain
    \begin{align*}
    &\Pr_{H, x_0}\left[y_j = w_{i_j} \middle | 
    \bigwedge_{\ell=1}^{j-1}y_\ell\neq w_{i_\ell}\right] \\
    & \le \Pr_{H, x_0}\left[y_j = w_{i_j} \middle | x_j\neq w_{i_j-1} \wedge
    \bigwedge_{\ell=1}^{j-1}y_\ell\neq w_{i_\ell}\right] \\
    &+ \Pr_{H, x_0}\left[x_j = w_{i_j-1} \middle |
    \bigwedge_{\ell=1}^{j-1}y_\ell\neq w_{i_\ell}\right].
    \end{align*}
    Notice that the first of the two events in the last sum can only occur if
    $x_j$ does not appear in $R$. Otherwise, $y_j\neq w_{i_j}$ due to the fact
    that none of the previous queries collided with $W$. Therefore, the reply
    $y_j$ is sampled uniformly at random, and this term is at most $\frac 1N$.
    
    To bound the second term, note that each previous reply $y_\ell$, provided 
    that 
    it does not collide with $W$, rules out at most one possible value for 
    $w_{i_j-1}$: either $x_\ell$  if $i_\ell = i_j$, or $y_\ell$ if $i_\ell = 
    i_j-1$. Therefore, $w_{i_j-1}$ is distributed uniformly over the remaining 
    values, of which there are at most $N-(j-1)$. Specifically $w_{i_j-1}$ is 
    equal 
    to $x_j$, which is a function of all the previous replies  $y_1,
    \dotsc,y_{j-1}$,
    with probability at most $\frac 1{N-j+1}$.
    
    Overall, 
    \begin{align*}
    \Pr_{H, x_0}\left[A \text{ loses}\right] &\ge \prod_{j=1}^{T+1}
    \left( 1-  \frac 1N - \frac 1{N-j+1}\right) \\
    &\ge \prod_{j=1}^{T+1}
    \left( 1-  \frac 2{N-j+1}\right). \\
    \end{align*}
    
    We note that this is a telescopic product, which simplifies to 
    \[
    \frac{(N-T-2)(N-T-1)}{N(N-1)} \geq \frac{N^2 - (2T + 3)N}{N^2}
    \]
    and therefore,
    
    \[
    \Pr_{H, x_0}\left[A \text{ wins}\right] \leq \frac{2T+3}{N}.
    \]
\end{proof}

Theorem~\ref{thm:onlineinvert} establishes the difficulty of finding a preimage
of the last iterate of the hash chain. For T/Key, we also need to bound the
success probability of attacks that ``guess'' a preimage of the entire chain.
\begin{corollary}\label{cor:onlineall}
    Let functions ${h_1,\dotsc,h_k \in [N]\rightarrow [N]}$ be chosen 
    independently and uniformly at random. Let $A$ be an algorithm that 
    gets oracle access to each of the functions $\{h_i\}_{i=1}^k$ and makes at 
    most 
    $T$ oracle queries overall. Then, 
    \begin{align*}
    \P_{\substack{h_1,\dotsc,h_k \in \mathcal F_N\\ x_0 \in [N]}}
    \left[h_{[1,k]}\left(A(h_{[1,k]}
    (x_0))\right)  = h_{[1,k]}(x_0)\right] \le \frac {2T+2k+1}{N}.
    \end{align*}
\end{corollary}
\begin{proof}
    Let $A$ be an algorithm as in the statement of the corollary. We use it to
    construct an algorithm $A'$ that finds a preimage of the last iterate of the
    hash chain (as in the statement of Theorem~\ref{thm:onlineinvert}). On input
    $y$, algorithm $A'$ runs algorithm $A$ to get a point $z$ and then computes
    and outputs $z'=h_{[1,k-1]}(z)$. If $h_{[1,k]}(z)=h_{[1,k]}(x)$, then
    $h_k(z')=h_{[1,k]}(x)$. Moreover, algorithm $A'$ makes at most $T'=T+k-1$
    queries to its oracles. Therefore by Theorem~\ref{thm:onlineinvert}, its
    success probability is at most $(2T'+ 3)/N=(2T+2k+1)/N$.
\end{proof}

\paragraph{Optimality.} 
One might ask whether the above lower bound is tight. Perhaps composing $k$
independent hash functions not only avoids some of the problems associated with
using a hash chain derived by composing the same hash function, but actually
results in a function that is $k$ times more difficult to invert than the basic
hash function. Ideally, one might have hoped that the probability of inverting
the hash chain in $T$ queries would be at most $O\left(\frac{T}{kN}\right)$. 

However, this is not the case, because every iteration of the hash chain
introduces additional collisions and shrinks the domain of the function at a
rate of $1/k$, where $k=o(N)$ is the length of the chain (see
Lemma~\ref{lem:shrinkage} for a proof sketch). An attacker can use these
collisions to her advantage. Consider an attack that evaluates the chain on
$T/k$ random points in its domain (at a total cost of $T$ hash computations).
Lemma~\ref{lem:preimages} shows that a point in the image of the hash chain has
$k$ preimages in expectation. Therefore, each of the $T/k$ randomly chosen
points collides with the input under the chain with probability $k/N$, and so the overall success probability
of the attack is roughly $T/N$.

\subsection{Security of T/Key}

Our threat model assumes the adversary can repeatedly gain access to the server
and obtain all information needed to verify the password. The adversary can also
obtain multiple valid passwords at times of his choice. Finally, we allow the
adversary to choose the time when he makes his impersonation attempt. To
mitigate preprocessing attacks, we salt all our hash functions (in
Section~\ref{sec:preprocessing}, we discuss preprocessing attacks in more
detail, including the extent to which salting helps prevent them).

\paragraph{Non-threats.}
First, we assume that there is no malware on the phone or on the
user's laptop.  Otherwise, the user's session can be hijacked by the
malware, and strong authentication is of little value.
Second, because the channel between the laptop and the
authentication server is protected by TLS, we assume there is no
man-in-the-middle on this channel.
Third, all TOTP schemes are susceptible to an {\em online} phishing
attack where the attacker fools the user into revealing her
short-lived one-time password to a phishing site, and the
attacker then immediately authenticates as the user, within the allowable
short window.  This is a consequence of the requirement for {\em one-way}
communication with the authentication token (the phone).  Note however
that the limited time window makes the exploitation of credentials 
time-sensitive, which makes the attack more complicated.

We begin by presenting a formal definition of security. Our definitions are 
based on standard definitions of identification protocols (see, for example,
\cite{Shoup_1999}).

\begin{definition}[Time-based One-Time Password Protocol]
A one-time password protocol is a tuple $\mathcal I=(\texttt{pp}, 
\texttt{keygen}, P, V)$ where
\begin{itemize}
    \item Public parameter generator $\texttt{pp}(1^\lambda, k) \rightarrow n$ 
    is a polynomial time algorithm that takes as input the security parameter 
    in unary along with the maximal supported 
    authentication period $k$ and outputs the password length $n$.
   \item Key generator $\texttt{keygen}(n,k) \rightarrow
   (\mathit{sk},\mathit{vst})$ is a probabilistic polynomial time algorithm that
   takes as input the parameters, $n$ and $k$, and outputs the prover's secret
   key $\mathit{sk}$ and the initial verifier state $\mathit{vst} $.
    \item Prover $\texttt{P}\,(\mathit{sk},t)\rightarrow p_t$ is a polynomial
    time algorithm, which takes as input the prover's secret key $\mathit{sk}$,
    and a time $t \in [1,k]$, and outputs a one-time password $p$.
    \item Verifier $\texttt{V}\,(\mathit{vst},p,t)\rightarrow
    (\texttt{accept/reject},\mathit{vst}')$ is a polynomial time algorithm,
    which takes as input the previous state $\mathit{vst}$, a password $p$, and
    time $t\in [1,k]$ and outputs whether the password is accepted and the
    updated verifier state $\mathit{vst}'$.
\end{itemize}
For correctness, we require that when executed on monotonically increasing 
values of $t$ with
the state $\mathit{vst}$ properly maintained as described above, the verifier 
$\texttt{V}\,(\mathit{vst}, 
\texttt{P}\,(\mathit{sk},t),t)$ always outputs $\texttt{accept}$.
\end{definition}

We now proceed to define the security game, where we use the random oracle 
model~\cite{randomoracle}.

\begin{attackgame}
    \label{def:game}
Let $\mathcal I$ be a time-based one-time password protocol, and let $\mathcal O
$ be a random oracle. Given a challenger and an adversary 
$A$, the attack game runs as follows:
\begin{itemize}
    \item \emph{Public Parameter Generation} -- The challenger generates \\ $n
    \leftarrow \texttt{pp}(1^\lambda,k)$.
    \item \emph{Key Generation Phase} -- The challenger generates \\
    ${(vk,\mathit{sk}) \leftarrow
    \texttt{keygen}^{\mathcal O}(n,k)}$, given access to the random oracle.
    \item \emph{Query Phase} -- The adversary runs the algorithm 
    $A$, which is given the verifier's initial state $\mathit{vst}$ as well as 
    the ability to issue the following types of (possibly adaptive)  queries:
    \begin{itemize}
        \item Password Queries: The adversary sends the challenger a time value
        $t$. \\ The challenger generates the password $p\leftarrow
        \texttt{P}\,^\mathcal{O}(t,\mathit{sk})$, feeds it to the verifier to
        obtain $(\texttt{accept}, \mathit{vst}')\leftarrow
        \texttt{V}\,^{\mathcal O}(t,\mathit{vst},p)$, updates the stored
        verifier state to $\mathit{vst}'$,  and sends $p$ to the adversary.
        \item Random Oracle Queries: The adversary sends the challenger a point
        $x$, and the challenger replies with $\mathcal O(x)$.
    \end{itemize}
    The above queries can be adaptive, and the only restriction is that the 
    values of $t$ for the password queries must be monotonically increasing.
    \item \emph{Impersonation attempt} -- The adversary submits an 
    identification attempt $(t_{\mathrm{attack}},p_{\mathrm{attack}})$, such
    that $t_{\mathrm{attack}}$ is greater than  all previously queried password
    values.  
\end{itemize}
We say that the adversary $A$ wins the game if $\,\texttt{V}^{\,\mathcal O}
(vst,p_{\mathrm{attack}},t_{\mathrm{attack}})$ 
outputs \texttt{accept}. 
We let $\textbf{Adv}_A(\lambda)$ denote the probability of the adversary
winning the game with security parameter $\lambda$, where the probability is 
taken over the random oracle as well as the randomness in the key 
generation phase. 
\end{attackgame}

We are now ready to prove that T/Key is secure. 
Specifically, given an adversary 
that makes at most $T$ queries, we establish an upper bound on the advantage the 
adversary can have in breaking the scheme. 
We note that no such result was previously known for the original S/Key 
scheme, and the key ingredient in our proof is Theorem~\ref{thm:onlineinvert}.

\begin{theorem}[Security of T/Key]
    \label{thm:security}
    Consider the T/Key scheme with password length $n$  and maximum 
    authentication period $k$.  Let $A$ 
    be an adversary attacking the scheme that makes at most $T$ random oracle 
    queries. Then, 
    \[
    \textbf{Adv}_A \le \frac {2T+2k+1}{2^n}.
    \]
\end{theorem}

\begin{proof}
    First, recall that our scheme uses a hash function $H:\{0,1\}^m
\rightarrow \{0,1\}^m$ to get $k$ functions $h_1, \dots,h_k : \{0,1\}^n
\rightarrow \{0,1\}^n$, where  ${h_i(x) = H(t_{\mathrm{init}} + k - i
\|\mathit{id} \|  x)
    \vert_n}$. In the random oracle model, we instantiate $H$ using the 
    random oracle, and so the resulting $k$ functions, $h_1, \hdots, h_k$, 
    are random and independent.  
        
    Without loss of generality, we assume that $t_{\mathrm{init}}=0$ and that
the latest password requested by the adversary is the top of the chain $p_k$
(since the functions $h_1,\dotsc,h_k$ are independent, any random oracle or
password queries corresponding to times earlier than the latest requested
password do not help the adversary to invert the remaining segment of the
chain).

By definition, the verifier accepts $(t_{\mathrm{attack}},p_{\mathrm{attack}})$
if and only if $h_{[k-t_{\mathrm{attack}}+1,k]}(p_{\mathrm{attack}}) =
h_{[1,k]}(\mathit{sk})$. Therefore, if the adversary wins the game, it must hold
that at least one query $q_j\in R$ collides with $W$. 
The proof then follows from Corollary~\ref{cor:onlineall}.
\end{proof}

\paragraph{Concrete Security} With this result at hand, we compute the 
password length required to make T/Key secure. For moderate values of $k$ (say, 
negligible in $2^n$), to make our scheme as secure as a $\lambda$-bit random 
function, it is enough to set $n=\lambda+2$, since then, assuming $k<T$,
\[
    \textbf{Adv}_A  \le \frac{2T+2k+1}{2^n} \le \frac{4T}{2^{\lambda+2}} = 
    \frac T{2^\lambda}\,.
\]  
For standard 128-bit security, we require passwords of length 130 bits. 

\section{Checkpointing for Efficient Hash Chain Traversal}

In our scheme, the client stores the secret $\mathit{sk}$, which is used as the
head of the hash chain. In the password generation phase, as described in
Section~\ref{sec:description}, the client must compute the value of the node
corresponding to the authentication time each time it wishes to authenticate. A
naive implementation would simply traverse the hash chain from the head of the
chain all the way to the appropriate node. Since T/Key uses long hash chains,
this approach could lead to undesirable latency for password generation. To
decrease the number of hashes necessary to generate passwords, the client can 
store several values (called ``pebbles") corresponding to various points in 
the chain.

There exist multiple techniques for efficient hash chain traversal using dynamic
helper pointers that achieve $O(\log n)$ computation cost per chain link with
$O(\log n)$ cells of storage~\cite{jakobsson,cj}. However, there are two key
differences between the goals of those schemes and our requirements. 
\begin{enumerate}
\item These techniques 
all assume sequential evaluation of the hash chain, whereas in our scheme, 
authentication attempts are likely to result in an access-pattern containing
arbitrary gaps. 
\item Previous schemes aim to minimize the overall time needed to take a single
step along the hash chain, which consists of two parts: the time needed to fetch
the required value in the hash chain, and the time needed to reposition the
checkpoints in preparation for fetching the future values. In our setting,
however, it makes sense to minimize only the time needed to fetch the required
hash value, potentially at the cost of increasing the time needed to reposition
the checkpoints. This is reasonable since the gaps between a user's
authentication attempts provide ample time to reposition the checkpoints, and it
is the time to generate a password that is actually noticeable to the user.
\end{enumerate}

If the user's login behavior is completely unpredictable, we can minimize the
worst-case password generation time by placing the checkpoints at equal
distances from one another. We call this the \emph{na\"ive} checkpointing
scheme. However, in many real-world scenarios, user logins follow some pattern
that can be exploited to improve upon the na\"ive scheme.

To model a user's login behavior, we consider a probability distribution that
represents the probability that the user will next authenticate at time $t$
(measured in units of time slots) given that it last authenticated at time $0$.
Additionally, we let each node in the hash chain be indexed by its distance from
the tail of the chain and let $\ell$ be the index of the head of the chain
(i.e., $\ell$ is the length of the remaining part of the hash chain). In this
model, valid future login times are the integers $\{1, 2, \hdots, \ell\}$, and
each node in the hash chain is indexed by the corresponding login time. By this,
we mean that the valid password at time $t$ is the value at node $t$. This
notation is illustrated in Fig.~\ref{chainfig}.

\begin{figure}[!h]
\centering
\resizebox{8cm}{!}{%
\begin{tikzpicture}
[place/.style={circle,draw,fill=blue!20,thick,
                 inner sep=0pt,minimum size=11mm}]
\node at (0,0)  [place]  (label0) [label=below:Tail] {0};
\node at (2,0)  [place]  (label1) {1};
\phantom{\node at (4,0)  [place] (label2) {};}
\phantom{\node at (4,0) [place] (label3) {};}
\node at (6,0) [place] (labelkm) {$\ell-1$};
\node at (8,0) [place] (labelk) [label=below:Head] {$\ell$};

\draw (4,0) node [align=center] {$\cdots$};

\draw[->,>=stealth',semithick] (label1) -- (label0) ;
\draw[->,>=stealth',semithick] (labelk) -- (labelkm) ;
\draw[->,>=stealth',semithick] (labelkm) -- (label3) ;
\draw[->,>=stealth',semithick] (label2) -- (label1) ;

\coordinate[below=1cm of label0] (t0);
\coordinate[below=1cm of labelk] (tk);
\draw[dashed,->,>=stealth',semithick] (t0) -- (tk) node[midway,above]{Time};

\end{tikzpicture}
}

\caption{The hash chain with time-labeled nodes.} \label{chainfig}
\end{figure}

The problem is then to determine where to place $q$ checkpoints, $0 \leq c_1
\leq c_2 \leq \hdots \leq c_q < \ell$, in order to minimize the expected
computation cost of generating a password.  We note that if the client
authenticates at time $t$ and $c_i$ is the closest checkpoint to $t$ with $c_i
\geq t$, then the computational cost of generating the password is $c_i - t$. If
no such checkpoint exists, then the cost is $\ell - t$. We do not take into account
the number of additional hash computations required to reposition the
checkpoints after generating a password. 

In order to make the analysis simpler, we relax the model from a ``discrete''
notion of a hash chain to a ``continuous'' one. By this, we mean that we make
the probability distribution modeling the client's next login time continuous
and allow the checkpoints to be stored at any real index in the continuous
interval $(0,\ell]$. Additionally, we allow authentications to occur at any real
time in $(0,\ell]$. Formally, let $p(t)$ be the probability density function (pdf)
of this distribution with support over the positive reals and let $F(t) =
\int_0^{t} p(t) dt$ be its cumulative distribution function (cdf). We can then
express the computational cost $C$ in terms of the checkpoints by the formula
\begin{align*}
C &= \int_0^{c_1} (c_1 - t) p(t)dt + \int_{c_1}^{c_2} (c_2 - t)p(t)dt + 
\hdots + \int_{c_q}^\ell (\ell - t)p(t)dt \\
&= c_1 F(c_1) + c_2 (F(c_2) - F(c_1)) + \hdots + \ell(F(\ell) - F(c_q)) -
\int_0^\ell t p(t) dt. \\
\end{align*}
In order to determine the values of the $c_i$'s that minimize $C$, we take the
partial derivatives $\frac{\partial C}{\partial c_i}$ for each variable and set
them equal to $0$. This gives the following system of equations:
\begin{align} \label{eqn:full_first}
\frac{F(c_1)}{p(c_1)} &= c_2 - c_1 \\
\frac{F(c_2) - F(c_1)}{p(c_2)} &= c_3 - c_2 \\
\vdots \\
\label{eqn:full_last}
\frac{F(c_q) - F(c_{q-1})}{p(c_q)} &= \ell - c_q\,.
\end{align}
Solving these equations yields the values of the $c_i$'s that minimize $C$,
which we then round to the nearest integer, since checkpoints can only be placed
at integer coordinates. We refer to this as the \emph{expectation-optimal}
solution.

Depending on the specific distribution, this system of equations may or may not
be numerically solvable. If necessary, one can simplify the problem by replacing
the set of dependent multivariate equations with a set of independent
univariate equations. This is done using the following recursive approach.
We first place a single checkpoint $c$ optimally in $[0,\ell]$, then place optimal
checkpoints in the subintervals $[0,c]$ and $[c,\ell$], and then place checkpoints
in the next set of subintervals, etc. The problem then reduces to the problem of
placing a single checkpoint in an interval $[a,b]$, and the optimal location $x$
can be determined by solving the equation
\begin{equation} \label{eqn:recursive}
\frac{F(x) - F(a)}{p(x)} = b - x\,.
\end{equation}

In practice, mobile second-factor devices are often not the best environment for
running numerical solvers. One solution would be to precompute the
expectation-optimal checkpoint positions for some fixed length $\ell$ (e.g., the
initial length of the chain) and distribution $F$ and then hardcode those values
into the second-factor application. However, as time progresses, these
precomputed positions will no longer be expectation-optimal for the the length
of the \emph{remaining} part of the hash chain. Moreover, one might want to
adaptively reposition the checkpoints based on the past average time between
logins of the user. 

\paragraph{Repositioning the checkpoints} Each time a password is generated, we
reposition the checkpoints by computing the optimal checkpoint positions for the
length of the remaining chain. We then compute the hash values at these
positions by traversing the hash chain from the nearest existing checkpoint.
This is done in the background after presenting the user with the generated
password.  

\subsection{User logins as a Poisson process}

One choice for modeling the distribution $F(t)$ between logins is the
exponential distribution 
\begin{align*}
p(t) = \lambda e^{-\lambda t} \quad \quad F(t)=  1- e^{-\lambda t}\,. 
\end{align*}
The exponential distribution is a distribution of the time between events in a
Poisson process, i.e. a process in which events occur continuously and
independently at a constant average rate. Previous works state that this is
a reasonable model for web login behavior~\cite{Blocki2013, Rasch}. In our
setting, the value of the average time between logins could vary anywhere
between hours and months depending on the specific application and whether a
second factor is required on every login, once in a period, or once per device.

For the exponential distribution, Equation~\ref{eqn:recursive} gives:
\[
\frac{-e^{-\lambda x} + e^{-\lambda a}}{\lambda e^{-\lambda x}} = b -x\,.
\] Conveniently, this equation admits the analytic solution
\begin{equation}
    \label{eqn:analytic} x = -\frac{W(e^{\lambda(x-a)+1})}{\lambda} + b +
1/\lambda\,,
\end{equation} 
where $W(\cdot)$ is the Lambert-W function~\cite{lambert}. The recursive
solution in this case can then be easily implemented on the second-factor
device.

Figure~\ref{fig:checkpointcomparison} compares the expected performance of the
following checkpointing procedures: na\"ive, expectation-optimal (obtained by
numerically solving Equations~\ref{eqn:full_first}-\ref{eqn:full_last}) and
recursive (obtained using Equation~\ref{eqn:analytic}). We also compare against
the pebbling scheme of Coppersmith and Jakobsson~\cite{cj}, although as we've
noted above, their scheme optimizes a different metric than ours, so it is no surprise that it does not perform as well as the recursive or expectation-optimal approaches in our setting.

\begin{figure}
    \includegraphics[width=\columnwidth]{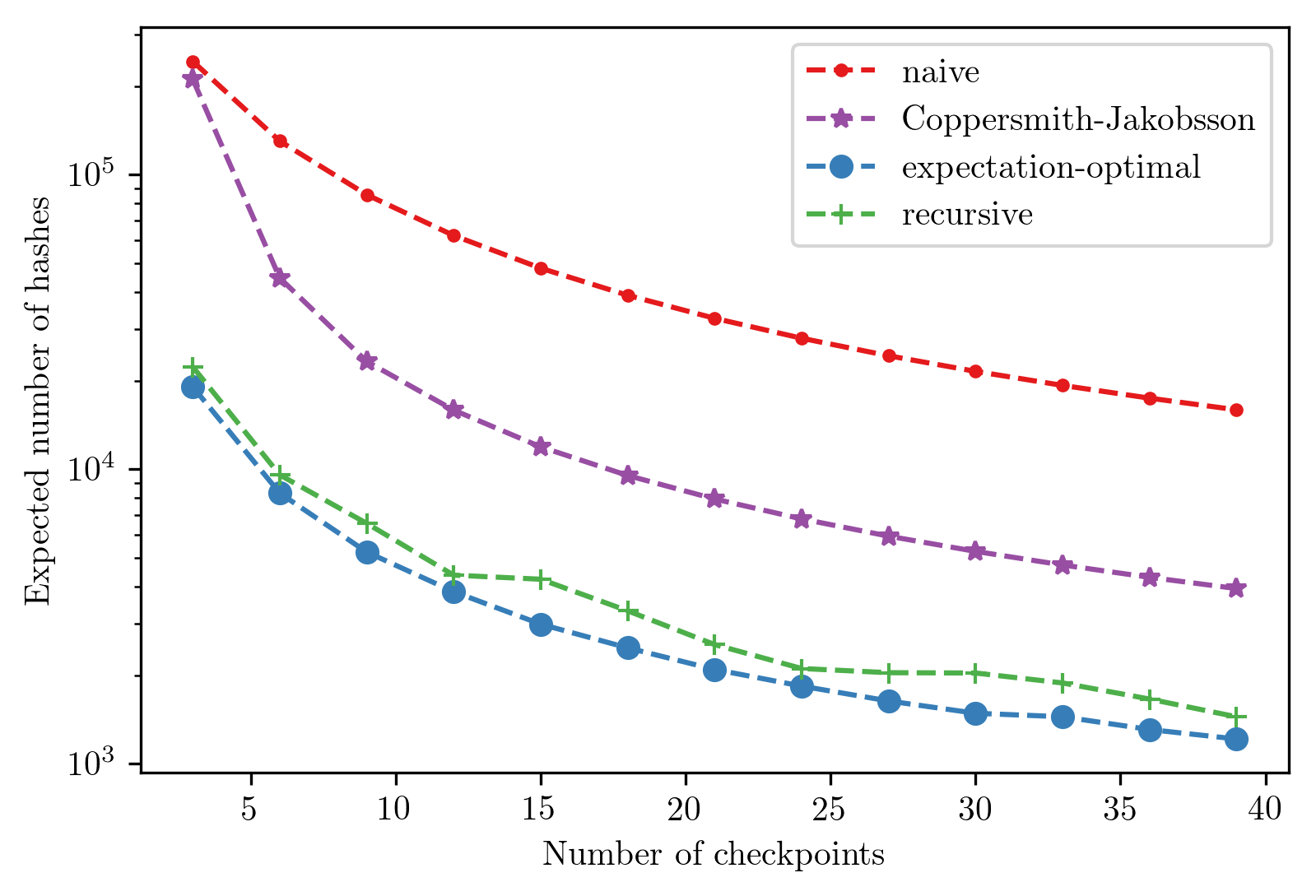}
    \caption{Performance of checkpointing schemes. Chain length
    is $1.05\times 10^6$ (one year when using $30$-second time slots). Login
    times are assumed to be a Poisson process with mean of $20160$ (one week
    when using $30$-second time slots). }
    \label{fig:checkpointcomparison}
\end{figure}

\paragraph{Balancing worst and expected performance} One disadvantage of both the
expectation-optimal and recursive checkpoints is that they perform poorly in the worst-case. Specifically, if a user does not
log in for a long period of time, a subsequent login might result in an
unacceptably high latency. A simple solution is to place several additional checkpoints in order to minimize the maximal distance between checkpoints, which  bounds the worst case number of hash computations. 

Figure~\ref{fig:checkpoint} illustrates the placement of checkpoints given by the different checkpointing schemes discussed in this section plotted along the probability density function of the exponential distribution.

\begin{figure}[!h]
    \includegraphics[width=\columnwidth]{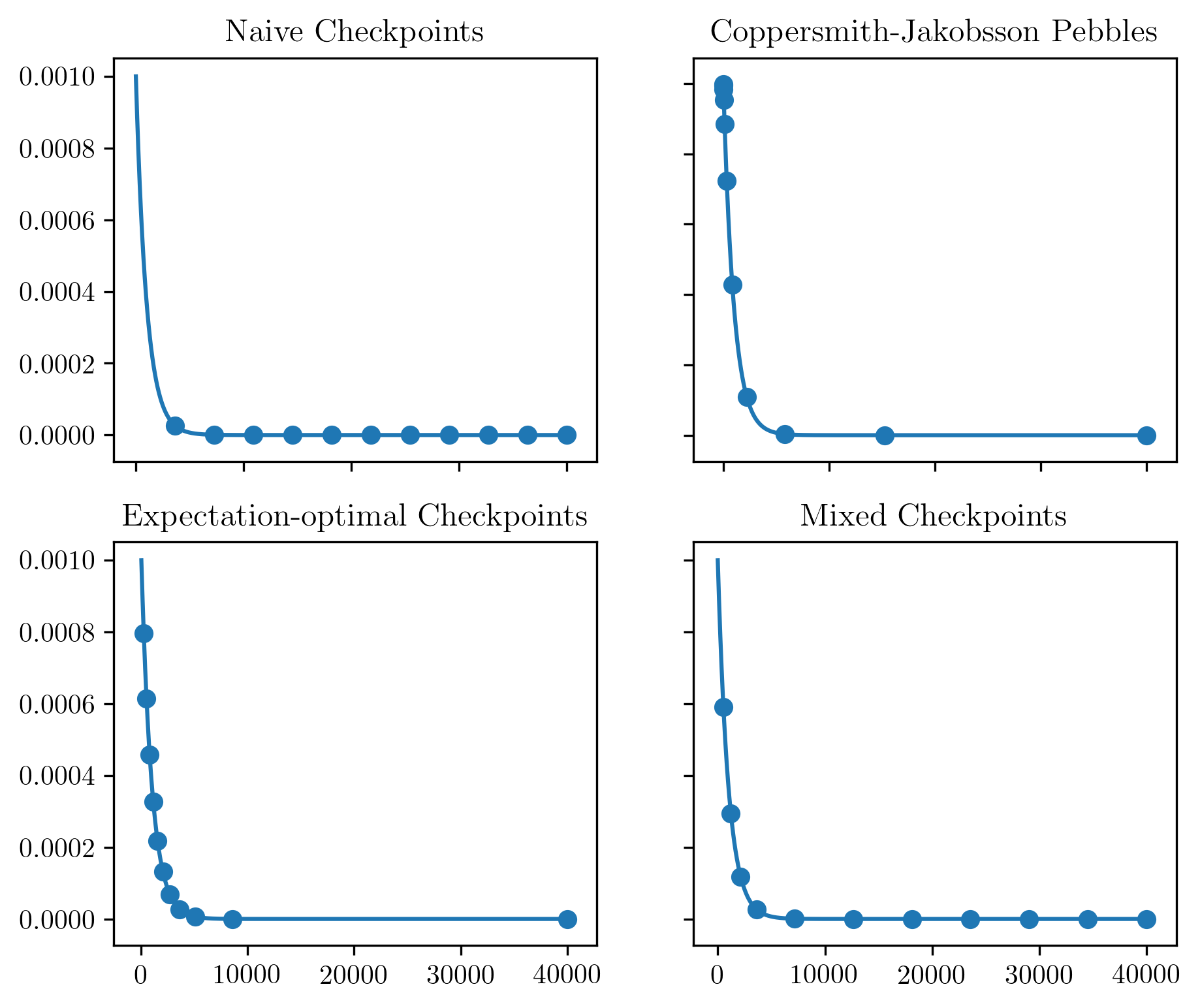}
    \caption{Illustration of different checkpointing schemes with logins modeled by the exponential
    distribution.} 
    \label{fig:checkpoint}
\end{figure}

\section{Implementation}
\label{sec:impl}
We implemented our scheme by extending the Google Authenticator Android App
and the Google Authenticator Linux Pluggable Authentication Module 
(PAM)~\cite{googleauth}.

\paragraph{Scheme Details and Parameters} We use passwords of length $130$ to
obtain the level of security discussed in Section~\ref{sec:security}. As a
concrete instantiation of a family of independent hash functions, for ${0\le i <
2^{32}}$, we take ${h_i:\{0,1\}^{130} \rightarrow 
\{0,1\}^{130}}$ to be defined as $h_i(x)=\texttt{SHA-256}(\langle i
\rangle_{32}\|\mathit{id} \| x)\vert_{130}$, where $\langle i \rangle_{32}$ is the
index of the function represented as a $32$-bit binary string, and $\mathit{id}$
is a randomly chosen $80$-bit salt. Our time-based counter uses time slots of
length $30$ seconds with $0$ being the UNIX epoch. The length of the hash chain
has to be chosen to balance the resulting maximal authentication period and the
setup time (which is dominated by the time to serially evaluate the entire  hash
chain). We use $2^{21}$ as our default hash chain length, resulting in a maximum
authentication period of approximately $2$ years and a setup time of less than
$15$ seconds on a modern mobile phone (see Section~\ref{sec:evaluation} for more
details). 

\subsection{Password Encoding}
Since the one-time passwords in our scheme are longer than those in the HMAC-based 
TOTP scheme (130 bits vs. 20 bits), we cannot encode the generated 
passwords as short numerical codes. Instead, we provide two encodings, which we 
believe are better suited for passwords of this length. 

\paragraph{QR Codes} First, our Android app supports encoding the one-time
password as a QR code. Among their many other applications, QR codes have been
widely used for second factor authentication to transmit information from the
authenticating device to the mobile device. For example, in Google
Authenticator, a website presents the user with a QR code containing the shared
secret for the TOTP scheme, which the user then scans with her mobile phone,
thus providing the authenticator app with the secret. QR codes have also been
used for transaction authentication as a communication channel from the insecure
device to the secure one~\cite{qrtan}.  

In our scheme, QR codes are used in the authentication process as a 
communication channel from the secure mobile device to the authenticating 
device. Such a use case was previously considered by~\cite{SJSN} and was shown to be practical~\cite{oneswipe}.  Specifically, our app encodes 
the 130 bit password as a QR code of size $21\times 21$ modules, which is then displayed to the user.
To log in on a different device, the user can then use that device's camera to scan 
the QR code from the mobile phone's screen. This method is best suited for use 
on laptops, tablets, and phones, where built-in cameras are ubiquitous, yet it 
can also be used on desktops with webcams. 
The QR code password encoding also provides a clear visualization of the 
relatively short length of our passwords compared to schemes using public key 
cryptography. For example, the standard ECDSA digital signature 
scheme~\cite{ecdsa} with a comparable level of security would result in 
$512$-bit long one-time passwords, which would consequently require larger 
$33\times 33$ QR codes~\cite{qrcodes} (a visual comparison appears in 
Figures~\ref{fig:ourqr} and \ref{fig:sigqr}). More recent digital signature 
constructions~\cite{bls, bernstein2012high} could be used to obtain shorter 
signatures, yet at $384$ and $256$ bits, respectively, those are still considerably longer than the one-time 
passwords in our scheme.

\begin{figure}[t] \label{fig:qr}
    \centering
    \ignore{
        \subfloat[Screenshot][QR Code password encoding in the mobile app]{
            \includegraphics[width=0.2\textwidth]{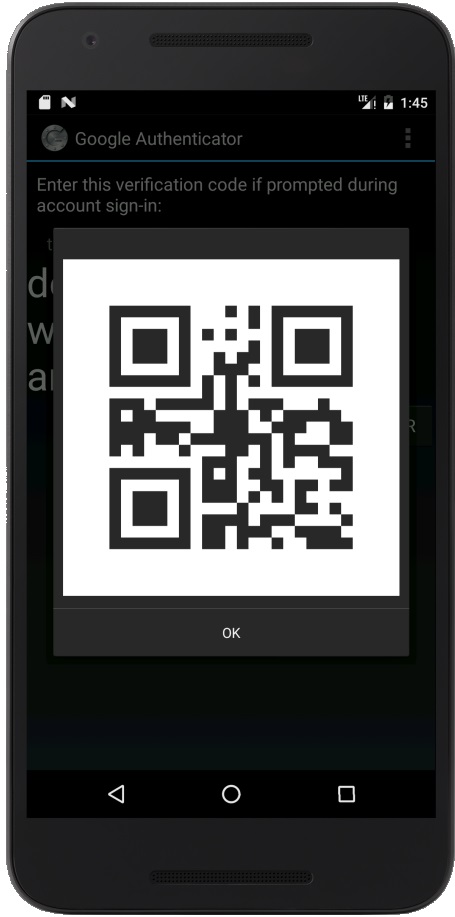}
            \label{fig:screenshot}
        }
        \qquad
    }
    \subfloat[Our Qr][21x21 QR encoding of a \\128-bit OTP]
    {\includegraphics[width=0.2\textwidth]
        {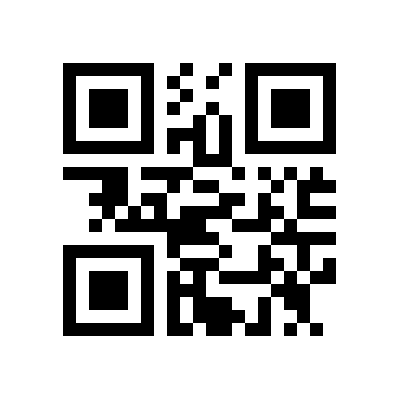}
        \label{fig:ourqr}}
    \subfloat[Our Qr][33x33 QR encoding of a 512-bit signature]
    {\includegraphics[width=0.2\textwidth]
        {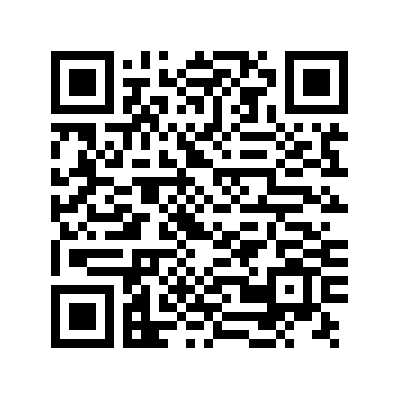}
        \label{fig:sigqr}}
    \caption{Password encoding using QR codes: T/Key vs. ECDSA signatures} 
    \label{fig:QR}
\end{figure}

\paragraph{Manual Entry} Since our usage of QR codes requires the sign-in device
to have a camera, we present an alternative method that can be used for devices
without cameras. In these instances, our Android app also encodes one-time
passwords using a public word list. Using a word list of $2048$ short words  (1
to 4 letters), as used in S/Key, results in 12-word passwords, and using a
larger $4096$ word list (of words up to 6 letters long), results in 11-word
passwords. Additionally, more specialized word lists such as those
in~\cite{wordlist} can be used if word lists that enable autofilling and error
correction are desired. These would be particularly useful if the sign-in device
was a mobile phone.

Alternatively, it would be possible to generate the one-time passwords as
arbitrary strings that the user would then manually enter. Assuming every
character in the strings has 6 bits of entropy (which is roughly the case for
case-sensitive alphanumeric strings), the resulting one-time passwords would be
strings composed of 22 characters. While typing these one-time passwords
manually would be cumbersome, they are at least somewhat practical, as opposed
to $512$ bit/$86$ character long digital signatures.

\paragraph{Hardware Authentication Devices} USB-based hardware authentication
devices, such as Yubikey~\cite{Yubikey} are often used instead of mobile phone
apps for generating TOTP passwords. They offer two main advantages: (i) after
the initial setup, the TOTP secret never has to leave the secure hardware, which
makes it more secure against client-side malware, and (ii) such authentication
devices are capable of emulating a keyboard and can ``type'' the generated
one-time passwords into the relevant password field when the user presses a
button on the device. However, hardware tokens do not protect the
TOTP secret on the server. Additionally, the registration phase is
still susceptible to malware since the TOTP secret needs to be loaded into the
hardware token. The newer FIDO U2F protocol~\cite{u2f} addresses these problems,
yet it requires specialized support by the browser and two-way communication.

Hardware authentication devices and T/Key could therefore be well-suited for
each other: the hardware device would generate the hash chain, store the secret,
and provide the server with the initial password. When the user needs to
authenticate, the hardware token would traverse the chain and generate the
one-time password. T/Key would provide the security against server-side hacks,
and the hardware token would provide the security against client-side hacks.
Moreover, the ability of the hardware token to automatically ``type'' the
password would address one of T/Key's main disadvantages, namely that the passwords
are too long for manual entry.  


\begin{table*}
    \caption{Scheme Performance.\\
    \small{130 bit long passwords, 30 second time slots, 20 mixed checkpoints.}}
    \label{table:results}
    \begin{tabular}{@{}
        l
        c
        S
        S[table-text-alignment = center]
        S[table-text-alignment = center]
        S[table-text-alignment = center]}
        \toprule {Auth. Period} & {Mean Time} & {Setup Time} &
         \multicolumn{2}{c}{Password Generation Time} & {Verification Time} \\
        & {Between Logins} & {(seconds)} & \multicolumn{2}{c}{(seconds)} &
        {(seconds)} \\
        \cmidrule(l){4-5} 
        &&& {average case} & {worst case} & \\
        \midrule 1 year & {1 week} & 7.5 & 0.3& 0.6 & 0.4 \\
         2 years  & {2 weeks} &  14 & 0.5& 0.9 & 0.8 \\
         4  years  & {1 month} & 28 & 0.8 & 1.6 & 1.6 \\
        \bottomrule
    \end{tabular}%
\end{table*}

\section{Evaluation}
\label{sec:evaluation}

We evaluated the performance of our scheme to ensure the running times of 
its different stages are acceptable for a standard authentication scenario. The 
client Android app was tested on a Samsung Galaxy S7 phone (SM-G930F) with a 
2.3 Ghz Quad-Core CPU and 4 GB of RAM. The server side Linux PAM module was 
tested
on a 2.6 Ghz i7-6600 CPU with 4 GB RAM running Ubuntu 16.04.

Our evaluation uses 130-bit passwords and hash chains of length one, two, and
four million, corresponding to one-year, two-year, and four-year authentication
periods when a new password is generated every $30$ seconds. We evaluate the
following times:
\begin{itemize}
    \item Client setup time: the time it takes for the mobile phone to first generate
    the salt and the secret and then traverse the entire hash chain to
    compute the initial password and create the registration QR code. 
    \item Client password generation time: the time to traverse the chain from
    the closest checkpoint. We present both the worst-case time, which
    corresponds to the maximal distance between two checkpoints, as well as the
    expected time, which we simulate with respect to several typical exponential
    distributions.
    \item Server verification time:  the time to traverse the entire chain on
    the server. This captures the longest possible period between logins. In
    practice, this time will be much shorter if the user logs in regularly. 
\end{itemize}

Results appear in Table~\ref{table:results}. In general, we view several seconds
as being an acceptable time for the initial setup and a sub-second time as
acceptable for both password generation and verification.

We attribute some of the differences between the hash chain traversal time on
the server and the traversal time on the phone to the fact that the former was
tested using native C code, whereas the latter was run using a Java App on the
mobile phone. 


\section{Attacks with preprocessing}
\label{sec:preprocessing} One limitation of the previously discussed security
model is that we do not allow the adversary's algorithm to depend on the choice
of the random function $h$. In practice, however, the function $h$ is \emph{not}
a random function, but rather some fixed publicly known function, such as
SHA-256. This means that the adversary could perhaps query the function prior to receiving a challenge and store some information about it that could be leveraged later. In
this section, we bound the probability of success of such an attack by $\left(
ST/N\right)^{2/3}$, where $N=2^n$ is the size of the hash function domain. To
mitigate the risk of such attacks, we show that by salting all hash functions
with a random salt of length $n$, we can bound the probability of success by
$(T/N)^{2/3}$ (assuming $S\le N$). 

More formally, an inverting attack with preprocessing proceeds as follows:
\begin{itemize}
    \item First, a pair of algorithms $(A_0,A_1)$ are fixed.
    \item Second, the function $h$ is sampled from some distribution (e.g., the
    uniform distribution over all random functions over some set). 
    \item Third, given oracle access to $h$ (which is now fixed),
    \emph{preprocessing} algorithm $A_0$ creates an advice string $st_{h}$.
    \item Finally, the \emph{online} algorithm $A_1$ is given the advice string
    $st_h$, oracle access to the same $h$, and its input $y=h(x)$. 
\end{itemize}
The complexity of an attack in this model is usually measured by the maximal
length in bits of the advice string $st_h$, which is referred to as the
``space'' of the attack and denoted by $S$, and the maximal number of oracle
queries of the algorithm $A_1$, which is often referred to as the ``time'' of
the attack and is denoted by $T$. Note that at least for lower bounds we: (i)
allow the preprocessing algorithm an unlimited number of queries to its oracle
and (ii) only measure the number of queries made by $A_1$, ignoring all other
computation.

The power of preprocessing was first demonstrated in the seminal work of
Hellman~\cite{hellman1980cryptanalytic}, who showed that with preprocessing,
one-way permutations can be inverted much faster than by brute force.
Specifically, Hellman showed that for every one-way permutation
$f:[N]\rightarrow [N]$ and for every choice of parameters $S,T$ satisfying
$T\cdot S \ge N$, there exists an attack with preprocessing which uses space $S$
and time $T$.  Hellman also gave an argument for inverting a random function
with time-space tradeoff $ T\cdot S^2 \ge N^2$. Subsequently, Fiat and
Naor~\cite{fiat1991rigorous} gave an algorithm that works for all functions. The
inversion algorithm was further improved when Oechslin~\cite{oechslin2003}
introduced rainbow tables and demonstrated how they can be used to break Windows
password hashes.

Yao~\cite{yao90} investigated the limits of such attacks and proved that
${S\cdot T\ge \Omega(N)}$ is in fact necessary to invert a random function on
every point in its image. Yao's lower bound was further extended in
\cite{gennarotrevisan, wee2005, de2010}, which showed that attacks that invert a
random function with probability $\epsilon$ must satisfy
$ST\ge\Omega(\epsilon N)$. Recently, Dodis, Guo, and Katz~\cite{DGK} extended
these results by proving that the common defense of \emph{salting} is effective
in limiting the power of preprocessing in attacks against several common
cryptographic primitives. Specifically, for one way functions, they show: 
\begin{theorem}[\cite{DGK}]
    \label{thm:dgk} Let $h:[M]\times[N]\rightarrow [N]$ be a random function.
  Let $(A_0,A_1)$ be a pair of algorithms that get oracle access to $h$ such
  that $A_0$ outputs an advice string of length $S$ bits, $A_1$ makes at most
  $T$ oracle queries, and
    \[
    \P_{h,\, m\in [M]\,,x \in [N]}
    \left[h\left(m,A_1^h(A^h_0, m, h(m,x))\right) = h(m,x)\right] = \epsilon\,. 
    \] Then, 
\[ T\left(1+\frac SM\right) \ge \tilde\Omega(\epsilon N)\,.
\]
\end{theorem}

The above result can be interpreted as stating that by using a large enough salt
space $M$ (e.g., taking $M=N$), one can effectively remove any advantage gained
by having an advice string of length $S\le N$. Here, we study the potential of
using salts to defeat attacks with preprocessing on hash chains.

Let $\mathcal{F}_{M,N}$ denote the uniform distribution over the set of all
functions from $[M]\times[N]$ to $[M]\times[N]$ such that for all $f\in 
\mathcal F_{M,N}$ and all $(s,x)\in [M]\times[N]$, $f(s,x)=(s,y)$.

\begin{theorem} 
    \label{THM:INVERTCHAIN} Let functions ${h_1,\dotsc,h_k \in
    \mathcal{F}_{M,N}}$ be chosen independently and uniformly at random, where
    ${k=o(\sqrt{N})}$. Let $(A_0,A_1)$ be a pair of algorithms that get oracle
    access to each of the functions $\{h_i\}_{i=1}^k$, such that $A_0$ outputs
    an advice string of length $S$ bits,  $A_1$ makes at most $T$ oracle
    queries, and
    \begin{equation*}
    \hspace{-2em}
    \label{eqn:inv_prob}
    \P_{\substack{h_1,\dotsc,h_k \in \mathcal{F}_{M,N} \\ m\in[M]\,x \in [N]}} 
    \left[h_{[1,k]}\left(m,A^h(A^h_0,h_{[1,k]}(m,x))\right) =
    h_{[1,k]}(m,x)\right] = \epsilon\,.
    \end{equation*}
Then, \[ T\left(1+\frac{S}{M}\right) \ge \tilde\Omega(\epsilon^{3/2}N).
\]
\end{theorem}

We prove this theorem in Appendix~\ref{appendix:chain}.

\paragraph{Optimality} We do not know whether the above loss in the dependence
on $\epsilon$ is optimal. It would be interesting to try to prove a stronger
version of the above bound by directly applying the techniques of~
\cite{gennarotrevisan}. Even in the setting of constant $\epsilon$, where one
looks for the optimal dependence between $S,T$ and $N$, we do not know of an
attack matching the above bound for arbitrary intermediate values of $T$ and $S$
(apart from the boundary scenarios $T=N$ or $S=\frac Nk$). Rainbow
tables~\cite{oechslin2003}, which are the best generic attack to invert random
functions, give $S\sqrt{2T}=N$. Since a hash chain is not a random function (it
has many more collisions in expectation), the expected performance of rainbow
tables in our case is far from obvious. For arbitrary (rather than random)
functions, the best known attacks~\cite{fiat1991rigorous} have higher complexity
$TS^2=q N^3$, where $q$ is the collision probability of the function. Finding
better attacks is an interesting open question.

\subsection{Security of T/Key against preprocessing} 
Within the context of T/Key, Theorem~\ref{THM:INVERTCHAIN} leaves a couple of
gaps from our goal to make the salted T/Key scheme as secure against attacks
with preprocessing as it is secure against attacks without preprocessing. First
it has suboptimal dependence on the success probability $\epsilon$. Note that if
one only wants to rule out attacks that succeed with constant success
probability (say $1/2$ or $0.01$), then this gap is immaterial in terms of its
impact on the security parameters. Second, the theorem currently bounds the
probability to invert the entire hash chain, whereas to use it in Attack
Game~\ref{def:game}, one needs to prove a stronger version in which the
attacker can invert a chain suffix of his choice. We leave these two gaps as two
open problems.


\section{Related Work}

For a discussion of the many weaknesses of static passwords, 
see~\cite{herley2012research}.
One-time passwords were introduced by Lamport~\cite{lamport81} and later 
implemented as S/Key~\cite{skey}. HOTP and TOTP were proposed in 
\cite{hotp} and \cite{totp}, respectively. For a review and comparison of 
authentication schemes, see~\cite{o2003comparing, bonneau2012quest}. 
Leveraging trusted handheld devices to improve authentication security was discussed in
~\cite{Balfanz} and~\cite{mannan2007using}. Two-factor authentication schemes were analyzed rigorously in ~\cite{SJSN}, 
which proposes a suite of efficient protocols with various usability and security tradeoffs.

\paragraph{Online Two-Factor Authentication}
A large body of work has been devoted to the online setting, where one allows 
bidirectional digital communication between the server and the second-factor device~\cite{mannan2007using,garriss2008trustworthy,phoneauth,u2f, 
duoprompt}. In this setting, secrets on the server can usually be avoided by 
using public-key cryptography. We especially call the reader's attention to the 
work of Shirvanian et al.~\cite{SJSN}, who study multiple QR-based protocols.
In one of their schemes, called ``LBD-QR-PIN,''  the 
mobile device generates a key pair and sends the public key to the 
server. Subsequently, on each authentication attempt, the server generates a 
random 128-bit challenge, encrypts it using the client's public key, and sends 
it to the authenticating device. The authenticating device encodes the challenge 
as a QR code, which the user then scans using his mobile device. The mobile 
device decrypts the challenge using its stored private key, computes a short 6 
digit hash of the challenge, and presents it to the user. The user then enters 
this 6 digit code on the authenticating device, which sends it to the server for 
verification. A big advantage of this scheme lies in the fact that the 
messages that the client sends are very short and can therefore easily be entered manually by the user.

\paragraph{Hash Chains} For an overview of hash chains and their applications,
see~\cite{HJP, cj, jakobsson, goyal}. In particular, Hu et al.~\cite{HJP}
provide two different constructions of one-way hash chains, the Sandwich-chain
and the Comb Skipchain, which enable faster verification. They are less suited
for our setting since skipping segments of the chain requires the prover to
provide the verifier with additional values (which would result in longer
passwords). Goyal~\cite{goyal} proposes a reinitializable hash chain, a hash
chain with the property that it can be securely reinitialized when the root is
reached. Finally,~\cite{cj, jakobsson} discuss optimal time-memory tradeoffs
for sequential hash chain traversal. On the theoretical side, statistical
properties of the composition of random functions were studied as early
as~\cite{rubin1953} and gained prominence in the context of population dynamics
in the work of Kingman~\cite{kingman1982coalescent}. The size of the image of a set under the
iterated application of a random function was studied by Flajolet and
Odlyzko~\cite{Flajolet1990} and later in the context of rainbow tables in
\cite{oechslin2003,Avoine}. The size of the image of a set under compositions of
independent random functions was studied by Zubkov and Serov~\cite{zubkov2015,
zubkov2017}, who provide several useful tail bounds, some of which we use in
Appendix~\ref{appendix:composition}.

\paragraph{Attacks with preprocessing} Time-space tradeoffs, which we use as our
model in Section~\ref{sec:preprocessing}, were introduced by Hellman
\cite{hellman1980cryptanalytic} and later rigorously studied by Fiat and
Naor~\cite{fiat1991rigorous}. The lower bound to invert a function in this model
was shown by Yao~\cite{yao90} and, subsequently, extended
in~\cite{gennarotrevisan, wee2005, de2010, DGK}. The work of Gennaro and
Trevisan~\cite{gennarotrevisan} was particularly influential due to its
introduction of the ``compression paradigm'' for proving these kinds of lower
bounds. More attacks in this model were shown by Bernstein and Lange~\cite{BL}.

\section{Conclusions}

We presented a new time-based offline one-time password scheme, T/Key, that
has no server secrets. Prior work either was not time-based, as in S/Key, or
required secrets to be stored on the server, as in TOTP. We implemented T/Key as
a mobile app and showed it performs well, with sub-15 second setup time and
sub-second password generation and verification. To speed up the password
generation phase, we described a near-optimal algorithm for storing checkpoints
on the client, while limiting the amount of required memory. We gave a formal
security analysis of T/Key by proving a lower bound on the time needed to break
the scheme, which shows it is as secure as the underlying hash function. 
We showed
that by using independent hash functions, as opposed to iterating the same
function, we obtain better hardness results and eliminate several 
security vulnerabilities present in S/Key. Finally, we studied the 
general question of hash chain security and proved a
time-space lower bound on the amount of work needed to invert a hash chain in
the random oracle model with preprocessing.

\begin{acks}
We thank David Mazi\`eres for very helpful discussions about this work.
This work is supported by NSF, DARPA, a grant from ONR, and the
Simons Foundation. Opinions, findings, and conclusions or
recommendations expressed in this material are those of the authors
and do not necessarily reflect the views of DARPA.
\end{acks}

\appendix

\section{Collisions in Random Functions}
\label{appendix:composition}
We first need to investigate some statistical properties of compositions of 
random functions. Starting with the work of Kingman~\cite{kingman1982coalescent}, 
the distribution of the image size $\left|h_{[1,k]}([N])\right|$, and 
specifically its convergence rate to $1$ (\cite{donnelly1991weak, 
dalal2002compositions}), was studied.
In our setting, we are more interested in the properties of $h_{[1,k]}$ for moderate values of $k$, and specifically, we assume $k=o\left(\sqrt N \right)$. 
\begin{lemma}\label{lem:shrinkage}
    Let $k,N\in \mathbb N$ such that $k=o\left(\sqrt{N}\right)$. Then,
    \[
        \Ex_{h_1,\dotsc,h_k}\left[|h_{[1,k]}([N])|\right] = O\left(\frac Nk\right).
    \]
\end{lemma}

\begin{proof}
   A formal proof can be found in~\cite{zubkov2015}. Here, we provide a brief sketch of the argument. Let
    \[
        \alpha_k=\Ex_{h_1,\dots,h_k}\left[\left|h_{[1,k]}([N])\right|/N\right]\,.
    \]
    Then,
    \[
        \alpha_{k+1}=\Ex_{h_{k+1}}\left[|h_{k+1}([\alpha_kN])|/N\right]\,.
    \]
    The last expression can be interpreted as a simple occupancy problem of independently throwing $\alpha_kN$ balls into $N$ bins and reduces to
    the probability that a bin is not empty:
    \[
    \alpha_{k+1}  = 1-(1-1/N)^{\alpha_kN}\,.
    \]
    For large $N$, we can make the approximation $(1 - 1/N)^N \approx 1/e$. Substituting this gives 
    \[
    \alpha_{k+1} = 1 - e^{-\alpha_k}
    \]
    and Taylor expanding the resulting expression gives the following approximation for the recursive relation:
    \[
    \alpha_{k+1}= 1 - (1-\alpha_k+\alpha_k^2/2 - 
    O(\alpha_k^3)) = \alpha_k-\alpha_k^2/2 + O(\alpha_k^3)\,.
    \]
    Plugging-in the guess $\alpha_k= 2/k + O(1/k^3)$ into the right hand side gives
    \begin{align*}
        \alpha_{k+1} &= 2/k - 2/k^2 + O(1/k^3) = 2(k-1)/k^2 + O(1/k^3) \\
        &= 2/(k+1) - 2/(k^2(k+1)) + O(1/k^3)  = 2/(k+1) + O(1/k^3)
    \end{align*}
    and therefore
    \[
    \alpha_k= 2/k + O(1/k^3)
    \]
    satisfies the recursive relation.
    
\end{proof}

We also need to estimate the probability that any two points in the domain collide under the hash function. We make use of the following lemma, due to ~\cite{zubkov2015}, and give a short proof here for completeness.

\begin{lemma}[\cite{zubkov2015}]
    Let $k,N\in \mathbb N$ such that $k=o\left(\sqrt{N}\right)$, and let 
    $x,x',x''\in [N]$ such that $x\neq x' \neq x''$. Then, 
    \begin{align*}
    &\Pr_{h_1,\dotsc,h_k}\left[ h_{[1,k]}(x) = h_{[1,k]}(x') \right] = 
    \frac kN - o\left(\tfrac 1N\right) \\
    &\Pr_{h_1,\dotsc,h_k}\left[ h_{[1,k]}(x) = h_{[1,k]}(x') = h_{[1,k]}(x'') 
    \right] = \frac {k(3k-1 + o_{\scriptscriptstyle N}(1))}{2N^2} .
    \end{align*}
\end{lemma}
\begin{proof}
    Observe that since $h_1,\dotsc,h_k$ are independent, the random variables $h_{[1,i+1]}(x)$ 
    and $h_{[1,i+1]}(x')$ are independent when conditioned 
    on $h_{[1,i]}(x) \neq h_{[1,i]}(x')$. Using this fact gives
    \begin{align*}
    &\Pr_{h_1,\dotsc,h_k}\left[h_{[1,k]}(x)\neq h_{[1,k]}(x')\right]\\
     &=        
    \prod_{i=1}^k\Pr_{h_i}\left[h_{[1,i]}(x)\neq 
    h_{[1,i]}(x')\middle| h_{[1,i-1]}(x)\neq
    h_{[1,i-1]}(x')\right] \\
    &= \prod_{i=1}^k \left(1-\frac 1N\right)  =
    \left(1-\frac 1N\right)^k     
    \end{align*}
    and subsequently,
    \begin{align*}
    &\Pr_{h_1,\dotsc,h_k}\left[ h_{[1,k]}(x) = h_{[1,k]}(x') \right] = 
    1-\left(1-\frac 1N\right)^k\\ &= 
    1-\left(1-\frac kN + O\left(\frac {k^2}{N^2}\right)\right) =
    \frac kN - o\left(\tfrac 1N\right).  
    \end{align*}
    To show the second statement of the lemma, we break down the probability of a 3-collision between 
    $x,x',x''$ by iterating through the different levels in the hash chain 
    where a collision between $x$ and $x'$ could occur. We have that
    \small
    \begin{alignat*}{2}
    &\mathmakebox[0pt][l]{\Pr\left[ h_{[1,k]}(x) = h_{[1,k]}(x') = h_{[1,k]}
        (x'')
        \right]}  \\
    =\,&\mathmakebox[0pt][l]{\Pr\left[ h_{[1,k]}(x) = h_{[1,k]}(x'') \middle| 
        h_{[1,k]}(x)= h_{[1,k]}(x') \right] \cdot   \Pr\left[
        h_{[1,k]}(x)= h_{[1,k]}(x') \right]}
    \\
    = &\sum_{i=0}^{k-1}\bigg(&&\Pr\left[h_{[1,k]}(x) = h_{[1,k]}(x'') \middle|
    \min\left\{i':h_{[1,i']}(x)= h_{[1,i']}(x')\right\}=i+1\right]
    \\ 
    &&& \cdot \Pr\left[\min\left\{i':h_{[1,i']}(x) = 
    h_{[1,i']}(x') \right\}=i+1 \right]\bigg) \\
    = &\mathmakebox[0pt][l]{\sum_{i=0}^{k-1} \left(1- \left(1-\frac 2N 
        \right)^i
        \cdot\left(1-\frac 1N\right)^{k-i}\right)\cdot\left(1-\frac 1N
        \right)^i\cdot \frac 1N} \\
    = &\mathmakebox[0pt][l]{\sum_{i=0}^{k-1} \left(1-\left(1-\frac{2i}{N} + o
        \left(\frac 1N\right)\right)\cdot\left(1-\frac {k-i}N + o\left(\frac 1N
        \right)
        \right)\right)\cdot\left(1-\frac{i}{N} + o\left(\frac 1N\right)\right)
        \cdot 
        \frac 1N} \\
    = &\mathmakebox[0pt][l]{\sum_{i=0}^{k-1} 
        \left(\frac {i+k}{N^2} +o\left(\frac 1 {N^2}\right)\right) 
        = \frac{3k^2-k}{2N^2} + k\cdot o\left(\frac 1{N^2}
        \right)}
    \end{alignat*}
    \normalsize
    as desired.
\end{proof}

The next lemma estimates, for any $x \in [N]$, the expected number of 
preimages of the point $h_{[1,k]}(x)$ and its variance. We use $I_A$ to denote the indicator variable of probability event $A$.
\begin{lemma} \label{lem:preimages_ex}
    Let $\left\{h_i\right\}_{i=1}^k\in \mathcal F_N$ be independent random
    functions, and let $L_j = \sum_{i=1}^N I_{h_{[1,k]}(i) = j}$ be a random
    variable of the number of different preimages under $h_{[1,k]}$ of $j\in[N]
    $. For every $x\in[N]$,
    \begin{align*}
        &\Ex_{h_1,\dotsc,h_k}\left[L_{h_{[1,k]}(x)}\right] = k + 1- o(1) \\
        &\Var_{h_1,\dotsc, h_k}\left[L_{h_{[1,k]}(x)}\right] = \frac12(k+1)^2\,.
    \end{align*}
\end{lemma}

\begin{proof}
    From the linearity of expectation and the previous lemma, we find that
    \begin{align*}
    \Ex_{h_1,\dotsc,h_k}\left[L_{h_{[1,k]}(x)}\right] 
    &= \Ex_{h_1,\dotsc,h_k}\left[ 
    \sum_{x'=1}^N I_{h_{[1,k]}(x)=h_{[1,k]}(x')} \right] \\
    &= \sum_{x'=1}^N
    \Ex_{h_1,\dotsc,h_k}\left[I_{h_{[1,k]}(x)=h_{[1,k]}(x')}\right] \\ 
    &= \sum_{x'=1}^N\Pr_{h_1,\dotsc,h_k}\left[h_{[1,k]}(x)=h_{[1,k]}(x')\right] 
    \\
    &= 1 + (N-1)\cdot\left(\frac kN - o(\tfrac 1N)\right) = k + 1 - o(1). 
    \end{align*}
    Additionally, 
    \small
    \begin{align*}
    &\Ex_{h_1,\dotsc,h_k}\left[L^2_{h_{[1,k]}(x)}\right] \\
    &= \Ex_{h_1,\dotsc,h_k}\left[ 
    \sum_{x'=1}^N \sum_{x''=1}^N I_{h_{[1,k]}(x)=h_{[1,k]}(x')} \cdot 
    I_{h_{[1,k]}
        (x)=h_{[1,k]}(x'')}\right] \\
    &= \Ex_{h_1,\dotsc,h_k}\left[ \sum_{x',x''=1}^N I_{h_{[1,k]}(x)=h_{[1,k]}
        (x')=h_{[1,k]}(x'')} \right] \\
    &= (N-1)(N-2) \cdot \Ex_{\substack{h_1,\dotsc,h_k\\ x,x',x''\text{ 
                different}}}
    \left[ I_{h_{[1,k]}(x)=h_{[1,k]}(x')=h_{[1,k]}(x'')}\right] \\
    &\qquad+ 3(N-1) \cdot 
    \Ex_{\substack{h_1,\dotsc,h_k\\ x \neq x'}}\left[ I_{h_{[1,k]}(x)=h_{[1,k]}
        (x')} 
    \right] + 1 \\
    &=(N-1)(N-2)\cdot \left(\frac {k(3k-1 + o_{\scriptscriptstyle N}(1))}{2N^2} 
    \right) + 3(N-1)\cdot \left(\frac kN - o\left(\tfrac 1N\right) \right) + 1 
    \\
    &= \tfrac 32 k^2  + \left(\tfrac 52 + o_{\scriptscriptstyle N}(1)\right)k + 
    1 
    \end{align*}
    \normalsize
    and thus
    \begin{align*}
    \Var_{h_1,\dotsc, h_k}\left[L_{h_{[1,k]}(x)}\right] &= \Ex_{h_1,\dotsc,h_k}
    \left[L^2_{h_{[1,k]}(x)}\right]  - \Ex_{h_1,\dotsc,h_k}\left[L_{h_{[1,k]}
        (x)}
    \right]^2 \\
    &= \tfrac 12 k^2 + \left(\tfrac 12 + o_{\scriptscriptstyle N}(1)
    \right)k \le \tfrac 12 (k+1)^2.
    \end{align*} 
\end{proof}

\begin{lemma} \label{lem:preimages}
    Let $\left\{h_i\right\}_{i=1}^k\in \mathcal F_N$ be independent 
    random 
    functions, and let $L_j = \sum_{i=1}^N I_{h_{[1,k]}(i) = j}$ be a 
    random 
    variable of the number of different preimages under $h_{[1,k]}$ of 
    $j\in[N]
    $.
    For every $x\in[N]$,
    \[
    \Pr_{h_1,\dotsc,h_k}\left[L_{h_{[1,k]}(x)} \ge 
    \frac {2k}{\sqrt{\epsilon}} \right] 
    \le \frac \epsilon 2.
    \]
\end{lemma}
\begin{proof}
    Applying Chebyshev's inequality, we obtain
    \small
    \begin{align*}
    &\Pr_{h_1,\dotsc,h_k}\left[L_{h_{[1,k]}(x)} \ge 
    \frac {2k}{\sqrt{\epsilon}} \right]  \\
       &\le \Pr_{h_1,\dotsc,h_k}
    \left[L_{h_{[1,k]}
        (x)} \ge     (k+1)+ \sqrt{\frac 2 \epsilon} \cdot \sqrt {\frac 12}(k+1) 
        \right] 
    \le 
    \frac 
    \epsilon 2 . \qedhere
    \end{align*}
    \normalsize
\end{proof}

\begin{lemma}
     Let $\left\{h_i\right\}_{i=1}^k\in \mathcal F_{M,N}$ be independent 
     identically random functions, and let $L_j = \sum_{i=1}^N 
     I_{h_{[1,k]}(i) = j}$ be a random variable of the number of different 
     preimages under $h_{[1,k]}$ of $j\in [M] \times[N]$.
    For every $(s,x)\in[M]\times[N]$,
    \[
    \Pr_{h_1,\dotsc,h_k}\left[L_{h_{[1,k]}(s,x)} \ge 
    \frac {2k}{\sqrt{\epsilon}} \right] 
    \le \frac \epsilon 2.
    \]
\end{lemma}

\begin{proof}
    Fix $s\in [M]$, and define $h_{i,s}(x)$ to be the last $n$ bits of
    $h_i(s,x)$. Applying the previous lemma to $\{h_{i,s}\}_{i=1}^k$ yields the 
    result.
\end{proof}

\section{Proof of Theorem~\ref{THM:INVERTCHAIN}}
\label{appendix:chain}

Our proof reduces the problem of inverting a random function to the problem of 
inverting a random hash chain.

Let $\mathcal D$ be a distribution over functions from $\mathcal X$ to 
$\mathcal X$ such that for every $x\in \mathcal X$
\begin{equation}
\label{eqn:property}
\Pr_{h_1,\dotsc,h_k\in\mathcal D}\left[L_{h_{[1,k]}(x)} \ge 
\frac {2k}{\sqrt{\epsilon}} \right] 
\le \frac \epsilon 2.
\end{equation}
Let $A$ be an oracle algorithm, which, given oracle access to 
$h_1,\dotsc,h_k\in \mathcal D$ and $S$ bits of advice, makes at most $T$ oracle 
queries to all of its oracles combined and successfully inverts with probability 
$\epsilon\in(0,1)$. 
For any choice of functions $\{h_i\}_{i=1}^k\in\mathcal D$, let 
\[
G_{h_1,
    \dots,h_k}=
\left\{\,x\in\mathcal X
: \left(A(h_{[1,k]}(x)) = x \right)\wedge 
\left(L_{h_{[1,k]}(x)} \le \frac {2k}{\sqrt \epsilon}\right)   
\,
\right\}
\] be 
the set of \emph{good points}, where we 
say that a point is good if $A$ outputs $x$ when 
executed on $h_{[1,k]}(x)$, and the point does not have many collisions  
under  $h_{[1,k]}$. Note that the first condition is stronger than 
the condition that $A$ merely inverts $h_{[1,k]}(x)$.   
Denote by $h_{[1,k]}(G_{h_1,\dots,h_k})$ the corresponding set of 
good images. 
Observe that the second condition above guarantees 
that each point in the image $h_{[1,k]}(G_{h_1, \dots,h_k})$ has at most $
\frac{2k}{\sqrt
    \epsilon}$  preimages under $h_{[1,k]}$. Using this observation and a union 
bound, 
we conclude that
\small{%
    \begin{align*}
    &\Pr_{\substack{h_1,\dotsc,h_k \\ x\in \mathcal X }} \left[x \in 
    G_{h_1,
        \dots,h_k} 
    \right]  \ge 
    \frac {\sqrt \epsilon}{2k} \cdot \Pr_{\substack{h_1,\dotsc,h_k \\ x
            \in \mathcal X 
    }} \left[h_{[1,k]}(x) \in  h_{[1,k]}(G_{h_1,\dots,h_k})\right] \\
    &\ge \frac {\sqrt \epsilon}{2k} \cdot  \left(
    \Pr_{\substack{h_1,\dotsc,h_k \\ x\in \mathcal X }} \left[A(h_{[1,k]}(x)) 
    \in 
    h^{-1}_{[1,k]}(h_{[1,k]}(x))\right] 
    \right. \\ &\hspace{3.7em}-\left.
    \Pr_{\substack{h_1,\dotsc,h_k \\ x\in \mathcal 
            X 
    }} 
    \left[L_{h_{[1,k]}
        (x)} 
    \ge \frac {2k}{\sqrt \epsilon}
    \right] \right).
    \end{align*}
}
\normalsize

Plugging in the probabilities given by the theorem hypothesis and Equation~
\ref{eqn:property}, we obtain
\begin{align*}
\Pr_{\substack{h_1,\dotsc,h_k \\ x\in \mathcal X }} \left[x \in G_{h_1,
    \dots,h_k} 
\right] 
&\ge \frac {\sqrt \epsilon}{2k} \cdot (\epsilon-\tfrac 
\epsilon 2)
\ge \frac {\epsilon^{3/2}}{4k}.
\end{align*}

For any $i\in[k]$, let $G^i_{h_1,\dotsc,h_k}$ be the subset of points
in $G_{h_1,\dotsc,h_k}$ on which $A$ queries its $i$-th oracle function at 
most $\frac{2T}{k}$ times. 
Note that for every input and every choice of hash functions, the total number 
of queries is at most $T$, and so for every input, $A$ queries at least $1/2$ 
of 
its oracle functions at most $\frac{2T}k$ times. Therefore
\small
\begin{align*}
&\Pr_{\substack{h_1,\dotsc,h_k \\ x\in \mathcal X \\i\in[k] }} \left[x \in 
G^i_{h_1,
    \dots,h_k} 
\right]    \\
&=\Pr_{\substack{h_1,\dotsc,h_k \\ x\in \mathcal X \\i\in[k] }}\left[x\in 
G^i_{h_1,\dotsc,h_k}\middle| x \in G_{h_1,
    \dotsc,h_k} \right] \cdot \Pr_{\substack{h_1,\dotsc,h_k \\ x\in \mathcal X 
}} 
\left[x \in G_{h_1,
    \dots,h_k} 
\right] \\ &\ge \frac {\epsilon^{3/2}}{8k}.
\end{align*}
\normalsize

Therefore, there exists some fixed index $i^*\in[1,k]$ and some fixed
choice of all  the other hash functions 
$h_1, \dotsc,h_{i^*-1},h_{i^*+1},\dotsc,h_k$ that achieves    
a probability of at least $\frac {\epsilon^{3/2}}{8k}$ over a random 
$h_{i^*}$ and a random ${x\in 
    G_{h_1,\dotsc,h_k}}$. For every function $h$, denote by $G^{i^*}_h$ 
the 
set 
$G^{i^*}_{h_1,\dotsc,h_k}$ with $h_{i^*}=h$ and the other functions 
fixed 
as 
above. We get
\begin{align}
\label{eqn:T_on_G}
\P_{\substack{h \in \mathcal D\\ x \in\mathcal X}}
\left[x \in G_h^{i^*}
\right] \ge \frac {\epsilon^{3/2}}{8k}.
\end{align}

The choice of $i^*$ as well as the explicit description of the 
functions
$h_1, \dotsc,h_{i^*-1},h_{i^*+1},\dotsc,h_k$ can be hard-coded into 
the 
algorithm $A$ since Theorem~\ref{thm:dgk}, and therefore also our 
reduction, can be arbitrarily non-uniform in the \emph{input size}. Another 
way of 
thinking 
about this is that since our model charges the algorithm only for 
oracle 
queries, an algorithm in this model can deterministically determine 
the 
best $i^*$ and the remaining functions by simulating $A$'s behavior 
on all 
possible inputs (without making any oracle queries). 

Consider the following algorithm $A'$ for inverting 
a random 
function $h\in \mathcal D$. Algorithm $A'$ gets the same $S$ bits of 
advice as 
$A$ and is given oracle access to $h$. On input $z\in\mathcal X$, $A'$ 
computes 
$y=h_{[i^*+1,k]}(z)$ and then simulates $A$ on $y$ as 
follows: $A'$ uses its own oracle to answer oracle queries to $h_{i^*}$ 
and 
uses the chosen functions $h_1, \dotsc,h_{i^*-1},h_{i^*+1},\dotsc,h_k$ to 
answer 
all other oracle queries. 
Furthermore, $A'$ bounds the number of queries to $h$ by $\frac {2T}k$. Thus, 
if 
during the simulation $A$ tries to make more than this number of queries to 
$h$, 
algorithm $A'$ aborts. Otherwise, $A'$ obtains $A(y)$ and then 
computes and 
outputs $h_{[1,i^*-1]}(A(y))$.

To analyze the success probability of $A'$, the key observation is 
that if $w\in h_{[1,i^*-1]}(G^{i^*}_h)$, then $A'$ 
inverts $h(w)$ successfully. To see this, note that if 
$w=h_{[1,i^*-1]}(x)$ 
for 
$x\in G_h^{i^*}$, then $A'$ simulates $A$ on 
\[
y=h_{[i^*+1,k]}(h(w))=h_{[1,k]}(x) \in G_h^{i^*},
\]
thus $A(y)=x$, and $A'(h(w)) = h_{[1,i^*-1]}(A(y)) =w $ as 
desired. 
Therefore
\begin{align*}
\Pr_{\substack{h\in\mathcal D\\ w\in \mathcal X 
}}
\left[A'(h(w)) \in h^{-1}(h(w))\right] 
&\ge \Pr_{\substack{h\in\mathcal D \\ w\in \mathcal X }} 
\left[w \in h_{[1,i^*-1]}(G') \right] \\
&= \Pr_{\substack{h\in\mathcal D \\ x\in \mathcal X }} 
\left[x \in G_h^{i^*}\right] \ge \frac {\epsilon^{3/2}}
{8k},
\end{align*}
where the penultimate equality holds because $h_{[1,k]}$ and 
therefore also 
$h_{[1,i^*-1]}$ have no collisions on $G_h^{i^*}$, and the last 
inequality follows from Equation~\ref{eqn:T_on_G}.

To complete the proof of the theorem, we apply the lower bound given by 
Theorem~\ref{thm:dgk} to algorithm $A'$ and the distribution $S_{M,N}$, which gives 
\[
\frac{2T}{k}\left(1 + \frac{S}{M}\right) \ge \tilde{\Omega}
\left(\frac{\epsilon^{3/2}N}{8k}\right),
\]
and therefore 
\[
T\left(1 + \frac{S}{M}\right) \ge \tilde{\Omega}(\epsilon^{3/2}N)
\]
as required. \qed

\bibliographystyle{ACM-Reference-Format}
\bibliography{chain}

\end{document}

%% file: intro.tex

Static passwords are notorious for their security
weaknesses~\cite{morris1979password, bonneau2010password,
    wiki:icloudhack, wiki:amhack, wiki:yahoohack}, driving
commercial adoption of two-factor authentication schemes, such as
Duo~\cite{duoprompt}, Google authenticator~\cite{google2sv}, and many
others.  Several hardware tokens provide challenge-response
authentication using a protocol standardized by the FIDO industry
alliance~\cite{u2f}.

Nevertheless, for desktop and laptop authentication, there is a strong
desire to use the phone as a second factor instead of a dedicated
hardware token~\cite{wu2004secure,duoprompt, McCune, SJSN}.  Several systems 
support
phone-based challenge-response authentication (e.g.,~\cite{duoprompt}), but
they all provide a fall back mode to a one-time password scheme.  The
reason is that challenge-response requires two-way communication with
the phone: uploading the challenge to the phone and sending the
response from the phone to the server.  However, one cannot
rely on the user's phone to always be connected.  When the user is
traveling, she may not have connectivity under the local cell
provider, but may still wish to use her laptop to log in at a hotel or to
log in using a workstation at an Internet Cafe.  In this case,
authentication systems, such as Duo, fall back to a standard timed-based
one-time password (TOTP) scheme.

Standard TOTP schemes~\cite{totp} operate using a shared key $k$
stored on both the phone and the authentication server.  The phone
displays a six digit code to the user, derived from evaluating
$\textsc{hmac}(k,t)$, where $t$ is the current time, rounded to the
current 30 second multiple.  This way, the code changes every 30
seconds and can only be used once, hence the name {\em one-time
    password}.  The user enters the code on her laptop, which sends it to
the server, and the server verifies the code using the same key~$k$.
The server accepts a window of valid codes to account for clock-skew.

The benefit of TOTP schemes is that they only require {\em one-way}
communication from the phone to the laptop, so they can function
even if the phone is offline (challenge-response requires two-way
communication with the phone and is mostly used when the phone is
online).  However, a difficulty with current TOTP is that the server
must store the user's secret key~$k$ {\em in the clear}.  Otherwise,
the server cannot validate the 6-digit code from the user.  With this
design, a break-in at the server can expose the second factor secret
for all users in the system.  A well-publicized event of this type is
the attack on RSA SecurID, which led to subsequent attacks
on companies that rely on SecurID~\cite{SecurIDattack}.

\paragraph{Our work.}
We introduce a TOTP system called T/Key that requires no secrets on
the server.  Our starting point is a classic one-time password system
called S/Key~\cite{skey}, which is not time-based and suffers from a
number of security weaknesses, discussed in the next section.  Our
work modernizes S/Key, makes it time-based (hence the name T/Key), and
addresses the resulting security challenges.

In T/Key, the phone generates a hash chain seed and uses this seed to 
construct a {\em long} hash chain, say of length two million, as 
depicted in Figure~\ref{sketch}. The phone encodes the tail of the 
chain $T$ in a QR code, which the user scans with her laptop and
sends to the authentication server for storage.  The phone then starts
at the element immediately preceding $T$ in the chain and walks one 
step backwards along the chain once every 30 seconds. It does so 
until it reaches the head of the chain, which is the seed.  At every 
step, the phone displays the current element in the chain, and the user 
logs in by scanning the displayed code on her laptop.  At the rate of 
one step every 30 seconds, a single chain is good for approximately two 
years, at which point the phone generates a new chain.  The details 
of the scheme are presented in Section~\ref{sec:description}.  As in
TOTP, there is only one-way communication from the phone to the 
laptop, and the phone can be offline. Moreover, a server compromise reveals 
nothing of value to the attacker.

\begin{figure}[t]
    \centering
    \resizebox{8cm}{!}{%
        \begin{tikzpicture}
        [place/.style={circle,draw,fill=blue!20,thick,
            inner sep=0pt,minimum size=11mm}]
        \node at (0,0)  [place]  (label0) [label=below:{$\substack{\text{Sent to 
        the server}\\ \text{at setup}}$}] {Tail};
        \node at (2,0)  [place]  (label1) [label=below:{$\substack{\text{Used to 
        } \\ \text{to authenticate}\\ \text{at }t=1}$}]  {};
        \phantom{\node at (4,0)  [place] (label2) {};}
        \phantom{\node at (4,0) [place] (label3) {};}
        \node at (6,0) [place] (labelkm)[label=below:{$\substack{\text{Used to 
                } \\ \text{to authenticate}\\ \text{at }t=k-1}$}] {};
        \node at (8,0) [place] (labelk) [label=below:Head] {Head};
        
        \draw (4,0) node [align=center] {$\cdots$};
        
        \draw[->,>=stealth',semithick] (label1) --  node[above] {$h_k$} 
        ++ (label0);
        \draw[->,>=stealth',semithick] (labelk)--  node[above] {$h_1$} 
        ++  (labelkm) ;
        \draw[->,>=stealth',semithick] (labelkm) --  node[above] {$h_{2}$} 
        ++  (label3) ;
        \draw[->,>=stealth',semithick] (label2) -- (label1) ;
        \end{tikzpicture}
    }
\caption{Sketch of T/Key}
\label{sketch}
\end{figure}
\begin{table*}
    \caption{A comparison of OTP schemes.\label{tbl:compare}}    
        \begin{tabular}{@{}lccS@{}}
            \toprule
            & No server secrets & 
            Time-varying passwords & {Password length in bits} \\
            &&& {(at $2^{128}$ security)} \\
            \midrule
            S/Key & \checkmark  & \xmark & {N/A}
            \\
            
            TOTP (HMAC) & \xmark & \checkmark  & 20  \\
            
            Digital Signatures (ECDSA/EdDSA) & \checkmark  & \checkmark & 512 \\
            \textbf{T/Key} & \checkmark & \checkmark & 130  \\
            \bottomrule
            \addlinespace[\belowrulesep]
        \end{tabular}%

    \small{Note: S/Key does not support this level of security.}
\end{table*}

Such a TOTP scheme presents a number of challenges.  First, security
is unclear.  Imagine an attacker breaks into the server and steals the
top of the chain $T$.  The attacker knows the {\em exact time} when
the millionth inverse of $T$ will be used as the second factor.  That
time is about a year from when the break-in occurs, which means that
the attacker can take a year to compute the millionth inverse of $T$.
This raises the following challenge: for a given $k$, how difficult
is it to compute the $k$-th inverse of $T$?   

If the same function is used throughout the entire hash chain, as in S/Key, 
the scheme is vulnerable to ``birthday attacks"~\cite{HJP} and is easier to break than the original hash function~\cite{hastad}.
A standard solution is to use a different hash function at every step
in the chain.  The question then is the following:  if $H$ is
the composition of $k$ {\em random} hash functions, namely
\[  H(x) \deq h_k(h_{k-1}(\cdots(h_2(h_1(x))) \cdots )),  \]
how difficult is it to invert $H$ given $H(x)$ for a random $x$ in the domain?
We prove a time lower bound for this problem in the random oracle model.
Additionally, given the possibility of making time-space tradeoffs in attacks 
against cryptographic primitives~\cite{hellman1980cryptanalytic, 
oechslin2003,de2010}, a natural follow up question is whether the scheme is still 
secure against offline attackers.
Building on the recent results of Dodis, Guo and Katz~\cite{DGK}, we prove a 
time-space lower bound for this problem that bounds the time to invert $H$, 
given a bounded amount of preprocessing space. As hash chains are a widely used primitive, we believe that our lower bounds, both with and without preprocessing, may be of independent interest.

From this security analysis, we derive concrete parameters for 
T/Key.  For
$2^{128}$ security, every one-time password must be 130 bits. Since entering these one-time passwords manually would be cumbersome, our phone implementation displays a QR
code containing the one-time password, which the user scans using her
laptop camera.  We describe our implementation in
Section~\ref{sec:impl} and explain that T/Key can be used as a drop-in
replacement for Google Authenticator. The benefit is that T/Key
remains secure in the event of a server-side compromise.

We also note that USB-based one-time password tokens, such as Yubikey~\cite{Yubikey},
can be set up to emulate a USB keyboard.  When the user presses the
device button, the token ``types'' the one-time password into a
browser field on the laptop.  This one-way communication setup is well
suited for T/Key: the token computes a T/Key one-time password and
enters it into a web page by emulating a keyboard.  Again,
this TOTP system remains secure in the event of a server-side
compromise.

\medskip
The second challenge we face is performance.  Because the hash chain
is so long, it is unreasonable for the phone to recompute the entire
hash chain on every login attempt, since doing so would take several seconds
for every login. Several amortized algorithms have been developed for
quickly walking backwards on a hash chain, while using little memory
on the phone~\cite{jakobsson,cj}.  The problem is that these schemes
are designed to walk backwards a {\em single} step at a time.  In our
case, the authenticator app might not be used for a month or, perhaps, even
longer.  Once the user activates the app, the app must quickly
calculate the point in the hash chain corresponding to the current
time.  It would take too long to walk backwards from the last login point,
one step at a time, to reach the required point.

Instead, we develop a new approach for pebbling a hash chain that
enables a quick calculation of the required hash chain elements.  We model
the user's login attempts as a Poisson process with parameter $\lambda$
and work out a near-optimal method to quickly compute the required points
with little memory on the phone.

\paragraph{Other approaches.}
T/Key is not the only way to provide a TOTP with no secrets on the
server.  An alternate approach is to use a digital signature.  
The phone maintains the signing key, while the server maintains the
signature verification key.  On every authentication attempt, the phone
computes the signature on the current time, rounded to a multiple of 30 seconds.
This can be scanned into the laptop and sent to the server to be verified
using the verification key.  

While this signature-based scheme has similar security properties to T/Key, it
has a significant limitation. Standard digital signatures such as
ECDSA~\cite{ecdsa} and EdDSA~\cite{bernstein2012high, rfc8032} are {\em 512 bits
long} for $2^{128}$ security\footnote{BLS signatures~\cite{bls} are shorter, but
require a pairing operation on the server which makes them less attractive in
these settings.}. 
These are about four times as long as the tokens used in T/Key. For example,
when encoded as QR codes, the longer tokens result in a denser QR code.
To preserve the maximal scanning distance, the denser QR code must be displayed 
in a larger image~\cite{qrsize}. (Alternatively, the signatures could be
decomposed into 
several QR codes of the original size, but scanning multiple images introduces
additional complexity for the user.) Short authentication tokens might also be 
desirable in other applications such as Bluetooth Low Energy (which supports a 
$23$-byte long MTU~\cite{bluetooth}).

Table~\ref{tbl:compare} provides a comparison of the
different TOTP mechanisms and their properties.  The last column
shows the required length of the one-time password.

\paragraph{Beyond authentication.}
Hash chains come up in a number of other cryptographic settings,
such as Winternitz one-time signatures~\cite{winternitz} and the Merkle-Damgard 
construction~\cite{Merkle}. Existing security proofs for Winternitz signatures 
often only take into 
account the attacker's online work.  Our lower bound on inverting
hash chains is well suited for these settings and can be used to 
derive time-space tradeoff proofs of security for these constructions. 
This is especially relevant as these schemes are being 
standardized~\cite{xmss,McGrew,Katz}.